\documentclass[
  aps,
  prl,
  reprint,
  superscriptaddress,
  nofootinbib,
  longbibliography
]{revtex4-2}
\usepackage{braket}
\usepackage[T1]{fontenc}
\usepackage[utf8]{inputenc}
\usepackage{lmodern}
\usepackage{amsmath,amssymb,mathtools}
\usepackage{microtype}
\usepackage[colorlinks=true,allcolors=blue]{hyperref}
\usepackage{dsfont}
\usepackage{amsthm}
\usepackage[colorlinks=true,allcolors=blue]{hyperref}
\usepackage{bm}
\usepackage{enumitem}

\newtheorem{theorem}{Theorem}
\newtheorem{definition}[theorem]{Definition}
\newtheorem{lemma}[theorem]{Lemma}

\newcommand{\Tr}{\operatorname{Tr}}
\newcommand{\PPT}{\text{PPT}}
\newcommand{\Sep}{\text{SEP}}

\newcommand{\one}{\mathds{1}}

\newcommand{\KK}{\mathcal{K}}
\newcommand{\KKP}{\mathcal{K}\textrm{P}}

\newcommand{\tracedist}[2]{T(#1,#2)}
\newcommand{\Eexpscc}{E_{c,\KKP}^{\exp(\dagger)}}

\begin{document}

\title{Very Strong Irreversibility of Quantum Entanglement}

\author{Tulja Varun Kondra}
\email{tuljavarun@gmail.com}
\author{Raphael Brinster}
\author{Hermann Kampermann}
\author{Dagmar Bruß}
\author{Nikolai Wyderka}
\affiliation{Institut für Theoretische Physik III, Heinrich-Heine-Universität Düsseldorf, Universitätsstraße 1, 40225 Düsseldorf, Germany}
\date{\today}

\begin{abstract}
The manipulation of quantum entanglement is fundamentally irreversible: some mixed entangled states require pure entanglement for their preparation, although no pure entanglement can be recovered from them by local operations and classical communication. This irreversibility is known to persist even under the maximal class of operations that do not generate entanglement, revealing a fundamental distinction between entanglement theory and thermodynamics. We construct cases for which any attempt to restore reversibility necessarily incurs an error that increases exponentially with the number of copies. Technically, we demonstrate a strict separation between the exponential strong-converse distillable entanglement and the exponential strong-converse entanglement cost. Our result resolves a conjecture posed by Lami and Regula (in \cite{Lami2023NoSecondLaw}) and strengthens it by showing that the irreversibility of entanglement persists even at the level of polynomially (in the number of copies) growing error. We further derive a semidefinite-programming lower bound on the exponential strong-converse cost under non-entangling operations. Finally, for the class of completely PPT-preserving operations, we construct analytically solvable families of antisymmetric states exhibiting the exponential strong-converse irreversibility. Remarkably, to our knowledge, no analogous separation between exponential strong converse cost and the analogous distillable entanglement is currently known even under the more restrictive class of LOCC operations.
\end{abstract}

\maketitle

\section{Introduction}
\label{sec:introduction}

Entanglement is one of the defining features of quantum physics~\cite{HorodeckiReview2009}. It gives rise to correlations that cannot be reproduced within any classical description~\cite{EinsteinPodolskyRosen1935,Schrodinger1935,Bell1964,ClauserHorneShimonyHolt1969}. Beyond its foundational significance, entanglement is an essential resource for quantum-information-processing tasks such as dense coding, quantum teleportation, entanglement-based cryptography, quantum repeaters, and measurement-based quantum computation~\cite{BennettWiesner1992,BennettTeleportation1993,Ekert1991,BriegelDuerCiracZoller1998,RaussendorfBriegel2001}. This operational role motivated the development of the resource theory of entanglement, in which spatially separated parties manipulate shared quantum states using local operations and classical communication (LOCC). LOCC comprises precisely those protocols in which each party performs local quantum operations and the parties coordinate their actions through classical communication, without exchanging quantum systems.

A natural unit of bipartite entanglement is the two-qubit maximally entangled state, or Bell pair, which serves as a ``golden unit'' in the resource theory of entanglement~\cite{BennettDiVincenzoSmolinWootters1996,HorodeckiReview2009}. This motivates two fundamental asymptotic tasks. Entanglement distillation asks for the largest rate at which Bell pairs can be extracted from many copies of a given state, whereas entanglement cost asks for the smallest rate of Bell pairs required to prepare many copies of that state, in both cases with a transformation error that vanishes in the asymptotic limit~\cite{BennettBernsteinPopescuSchumacher1996,BennettDiVincenzoSmolinWootters1996,HaydenHorodeckiTerhal2001}. A central discovery was the existence of bound entangled states, from which no pure entanglement can be distilled at any nonzero asymptotic rate, although their preparation requires a strictly positive rate of pure entanglement~\cite{Horodecki1997,HorodeckiBoundEntanglement1998,BennettUPB1999,VidalCirac2001,YangHorodeckiHorodeckiSynak2005}. The resulting gap between distillable entanglement and entanglement cost is the operational signature of irreversibility: distilling a state into Bell pairs and subsequently using those Bell pairs to recreate it necessarily recovers only a fraction of the initial number of copies.

This irreversibility stands in sharp contrast to the reversible structure of asymptotic thermodynamics. In axiomatic thermodynamics, equilibrium states are ordered by an extensive entropy function~\cite{LiebYngvason1999}, while in resource-theoretic formulations at fixed temperature the asymptotic conversion rates are governed by the nonequilibrium free energy, or equivalently by the relative entropy with respect to the Gibbs state,
\begin{equation}
F_\beta(\rho)-F_\beta(\gamma_\beta)
=
\beta^{-1}D(\rho\Vert\gamma_\beta).
\end{equation}
Under Gibbs-preserving operations, the relative entropy to the Gibbs state determines both the asymptotic formation and extraction rates, rendering the theory reversible~\cite{BrandaoEtAlThermal2013,BrandaoEtAlSecondLaws2015}. This contrast raised the question of whether the irreversibility of mixed-state entanglement is an intrinsic property of the resource or merely a consequence of the restrictive structure of LOCC. A natural axiomatic alternative is the class of non-entangling operations, namely CPTP maps that send every separable state to a separable state. Lami and Regula answered this question negatively by constructing a mixed entangled state whose distillable entanglement is strictly smaller than its entanglement cost even under this maximal class of exactly non-entangling operations~\cite{Lami2023NoSecondLaw}.

The conventional notion of irreversibility assumes that the transformation error vanishes in the asymptotic limit. This leaves open the possibility that reversibility could be recovered by allowing a fixed nonzero error. The relevant quantities are then the strong-converse distillable entanglement and the strong-converse entanglement cost, which characterize the limiting rates beyond which distillation or formation necessarily fails with asymptotically maximal error. Finite-error reversibility would require these two thresholds to coincide. For the state introduced by Lami and Regula, the distillation threshold is strictly smaller than the known formation rate, and they conjectured that the formation cost remains unchanged for every fixed error smaller than one. If true, this would show that the irreversibility of entanglement persists even when arbitrarily large, but nonmaximal, errors are allowed~\cite{Lami2023NoSecondLaw}.

In this work, we show that allowing a fixed nonzero error does not restore the reversibility of entanglement manipulation. For the state introduced by Lami and Regula, we determine both strong-converse thresholds and prove that the distillable entanglement remains strictly smaller than the entanglement cost for every fixed error below one, thereby resolving their conjecture. We establish the stronger exponential statement that any protocol operating beyond the corresponding distillation or formation threshold has a transformation fidelity that vanishes exponentially with the number of copies, or equivalently, an error that approaches its maximal value exponentially fast. We further derive an efficiently computable semidefinite-programming lower bound on the exponential strong-converse entanglement cost under non-entangling operations. Finally, we extend our analysis to completely PPT-preserving operations, an axiomatic class defined by complete preservation of the set of PPT states and strictly containing LOCC, and obtain a general criterion for exponential strong-converse irreversibility. Applying this criterion, we construct analytically solvable families of states whose formation cost is exactly one ebit, while their strong-converse distillation rate is strictly smaller.

\section{Definitions and background}
\label{sec:definitions}
We consider finite-dimensional bipartite systems $A$ and $B$. In the resource
theory of entanglement, the free states are the separable states,
\begin{equation}
  \operatorname{SEP}(A{:}B)
  :=
  \left\{
    \sum_x p_x\,\rho_A^x\otimes\rho_B^x
    :
    p_x\geq0,\ \sum_xp_x=1
  \right\}.
  \label{eq:separable-state}
\end{equation}
The operationally motivated transformations are local operations and classical
communication (LOCC), under which the parties may perform arbitrary local
quantum operations and coordinate their actions by exchanging classical
information, but cannot transmit quantum systems. Since every LOCC protocol
maps separable states to separable states, it cannot generate entanglement from
a free input
\cite{BennettDiVincenzoSmolinWootters1996,ChitambarEtAl2014,
ChitambarGour2019}.

To distinguish intrinsic limitations of entanglement from restrictions
specific to LOCC, it is natural to enlarge the allowed transformations while
retaining the defining resource-theoretic requirement that free states remain
free. This leads to the class of non-entangling operations,
\begin{equation}
  \operatorname{NE}
  :=
  \left\{
    \Lambda\ {\rm CPTP}:
    \Lambda(\operatorname{SEP})
    \subseteq\operatorname{SEP}
  \right\},
  \label{eq:non-entangling-definition}
\end{equation}
which is the maximal class of quantum channels that exactly preserve
separability and strictly contains LOCC \cite{ChitambarEtAl2014,ChitambarGour2019,Lami2023NoSecondLaw}. More generally, we consider a class of positive, trace-preserving
$\mathcal K$-preserving maps, denoted by $\mathcal K{\rm P}$, which satisfy
$\Lambda(\mathcal K)\subseteq\mathcal K$ for
$\mathcal K\in\{\operatorname{SEP},\operatorname{PPT}\}$. Here
\begin{equation}
  \operatorname{PPT}
  :=
  \left\{
    \rho\geq0:\rho^\Gamma\geq0
  \right\},
\end{equation}
where $\Gamma$ denotes partial transposition on Bob's subsystem.

A different tractable relaxation of LOCC is obtained by requiring preservation of the
PPT property even in the presence of arbitrary ancillary systems. A channel
$\Lambda_{AB\to A'B'}$ is completely PPT-preserving (C-PPT-P) if 
\begin{equation}
  \Gamma_{B'}\circ\Lambda\circ\Gamma_B
\end{equation}
is completely positive (here $\Gamma_B$ denotes the transposition of B). Equivalently, $\Lambda$ remains PPT-preserving under
extension by the identity channel on any ancillary bipartite system. Completely
PPT-preserving operations therefore form an axiomatic class based on complete
preservation of the PPT cone; they strictly contain LOCC and admit a useful semidefinite characterization. The classes of C-PPT-P and non-entangling (NE) maps are incomparable. The swap operation is NE but not C-PPT-P, whereas a channel that prepares a fixed PPT-entangled state is C-PPT-P but not NE.

The basic operational tasks are to extract a standard unit of entanglement from
a given state and to reconstruct the state from that unit. We take the standard
unit to be the maximally entangled state of Schmidt rank $M$,
\begin{equation}
  \Phi_M
  :=
  \lvert\Phi_M\rangle\!\langle\Phi_M\rvert,
  \qquad
  \lvert\Phi_M\rangle
  :=
  \frac{1}{\sqrt M}
  \sum_{i=0}^{M-1}\lvert ii\rangle ,
  \label{eq:maximally-entangled-state}
\end{equation}
which contains $\log_2 M$ ebits; in particular, $\Phi_2$ is a Bell
pair~\cite{BennettDiVincenzoSmolinWootters1996,HorodeckiReview2009}.
Throughout, all logarithms are base two, and approximation errors are measured
by the trace distance
\(
T(\rho,\sigma):=\frac12\lVert\rho-\sigma\rVert_1
\).

Let $\mathcal F$ denote the allowed class of operations. A distillation
protocol transforms $\rho^{\otimes n}$ into $\Phi_{M_n}$, whereas a formation
protocol implements the reverse conversion. We write
\begin{equation}
  r_n:=\frac{1}{n}\log_2 M_n
  \label{eq:conversion-rate}
\end{equation}
for the conversion rate, and denote the corresponding errors by
\begin{align}
  \varepsilon_n^d
  &:=
  T\!\left(
    \Lambda_n^d(\rho^{\otimes n}),
    \Phi_{M_n}
  \right),
  \nonumber\\
  \varepsilon_n^c
  &:=
  T\!\left(
    \Lambda_n^c(\Phi_{M_n}),
    \rho^{\otimes n}
  \right).
  \label{eq:conversion-errors}
\end{align}
Here $\Lambda_n^d,\Lambda_n^c\in\mathcal F$ denote the distillation and
formation protocols, respectively.

The fixed-error distillable entanglement is the largest asymptotic rate
achievable with error at most $\varepsilon$,
\begin{equation}
  E_{d,\mathcal F}^{\varepsilon}(\rho)
  :=
  \sup
  \left\{
    \liminf_{n\to\infty}r_n:
    \limsup_{n\to\infty}\varepsilon_n^d
    \leq\varepsilon
  \right\},
  \label{eq:fixed-error-distillation}
\end{equation}
whereas the fixed-error entanglement cost is the smallest rate sufficient for
formation,
\begin{equation}
  E_{c,\mathcal F}^{\varepsilon}(\rho)
  :=
  \inf
  \left\{
    \limsup_{n\to\infty}r_n:
    \limsup_{n\to\infty}\varepsilon_n^c
    \leq\varepsilon
  \right\}.
  \label{eq:fixed-error-cost}
\end{equation}
The optimizations are over all admissible protocol sequences and integers
$M_n$. The conventional asymptotic quantities are recovered in the
vanishing-error limit \cite{BennettDiVincenzoSmolinWootters1996,BrandaoDatta2011,
Lami2023NoSecondLaw},
\begin{equation}
  E_{d,\mathcal F}:=E_{d,\mathcal F}^{0},
  \qquad
  E_{c,\mathcal F}:=E_{c,\mathcal F}^{0}.
  \label{eq:vanishing-error-rates}
\end{equation}

A gap between these vanishing-error rates does not exclude the possibility
that reversibility is recovered when a fixed nonzero error is allowed. This
possibility is captured by the strong-converse thresholds
\begin{align}
  E_{d,\mathcal F}^{\dagger}(\rho)
  &:=
  \sup_{0\leq\varepsilon<1}
  E_{d,\mathcal F}^{\varepsilon}(\rho),
  \nonumber\\
  E_{c,\mathcal F}^{\dagger}(\rho)
  &:=
  \inf_{0\leq\varepsilon<1}
  E_{c,\mathcal F}^{\varepsilon}(\rho).
  \label{eq:strong-converse-thresholds}
\end{align}
Any distillation rate above
$E_{d,\mathcal F}^{\dagger}(\rho)$ forces the error to converge to one.
Likewise, any formation rate below
$E_{c,\mathcal F}^{\dagger}(\rho)$ forces the error to converge to one.
Consequently,
\begin{equation}
  E_{d,\mathcal F}^{\dagger}(\rho)
  <
  E_{c,\mathcal F}^{\dagger}(\rho)
  \label{eq:sc-irreversibility}
\end{equation}
certifies irreversibility even when an arbitrarily large but nonmaximal fixed
error is permitted. Such an exponentially strong converse irreversibility has not been shown before, even for the case of LOCC operations.

We strengthen this notion by asking how rapidly the error approaches one. A
rate $R$ is an exponential strong-converse rate for distillation if, for every
$R'>R$, there exists $\gamma>0$ such that every admissible protocol satisfying
$\liminf_{n\to\infty}r_n\geq R'$ obeys
\begin{equation}
  1-\varepsilon_n^d
  \leq
  2^{-\gamma n}
  \label{eq:exp-sc-distillation}
\end{equation}
for all sufficiently large $n$. The infimum of all such $R$ is the
exponential strong-converse distillable entanglement
$E_{d,\mathcal F}^{\exp(\dagger)}(\rho)$. Similarly, a rate $R$ is an exponential strong-converse rate for formation
if, for every $R'<R$, there exists $\gamma>0$ such that every admissible
protocol satisfying $\limsup_{n\to\infty}r_n\leq R'$ obeys
\begin{equation}
  1-\varepsilon_n^c
  \leq
  2^{-\gamma n}
  \label{eq:exp-sc-formation}
\end{equation}
for all sufficiently large $n$. The supremum of all such $R$ is the
exponential strong-converse entanglement cost
$E_{c,\mathcal F}^{\exp(\dagger)}(\rho)$. For more details we refer the reader to the Supplemental Material. Our main contribution will be to establish
\begin{equation}
  E_{d,\mathcal F}^{\exp(\dagger)}(\rho)
  <
  E_{c,\mathcal F}^{\exp(\dagger)}(\rho)
  \label{eq:exp-sc-irreversibility},
\end{equation}
for both $\mathcal{F}=\mathrm{NE}$ and $\mathcal{F}=\mathrm{PPTP}$, showing irreversibility at the exponential strong-converse level.

The natural test case is the two-qutrit state introduced by Lami and
Regula~\cite{Lami2023NoSecondLaw},
\begin{equation}
  \omega_3
  :=
  \frac{P_3-\Phi_3}{2},
  \qquad
  P_3
  :=
  \sum_{i=0}^{2}
  \lvert ii\rangle\!\langle ii\rvert .
  \label{eq:lami-regula-state}
\end{equation}
It can equivalently be written as a uniform mixture of three maximally
entangled two-dimensional states,
\begin{equation}
  \omega_3
  =
  \frac{1}{3}
  \sum_{0\leq i<j\leq2}
  \lvert\psi_{ij}\rangle\!
  \langle\psi_{ij}\rvert ,
  \qquad
  \lvert\psi_{ij}\rangle
  :=
  \frac{\lvert ii\rangle-\lvert jj\rangle}{\sqrt2}.
  \label{eq:omega-decomposition}
\end{equation}
For positive trace-preserving $\mathcal K$-preserving maps, with
$\mathcal K\in\{\operatorname{SEP},\operatorname{PPT}\}$, they established
\begin{equation} 
  E_{d,\mathcal K{\rm P}}(\omega_3)
  =
  \log_2\frac{3}{2}
  <
  1
  =
  E_{c,\mathcal K{\rm P}}(\omega_3).
  \label{eq:lami-regula-gap}
\end{equation}
Thus, the maximal resource-preserving classes considered in Ref.~\cite{Lami2023NoSecondLaw} do not restore reversibility in the vanishing-error regime. For the state $(\omega_3)$, Lami and Regula showed that the strong-converse distillable entanglement coincides with the usual distillable entanglement. They conjectured that the analogous equality also holds for the entanglement cost, namely that the strong-converse cost of $(\omega_3)$ remains equal to its usual one-ebit cost.

\section{Main results}
\label{sec:main-results}

We first resolve the finite-error conjecture of
Ref.~\cite{Lami2023NoSecondLaw}. The key ingredient is a single-copy,
semidefinite-programming bound on the exponential strong-converse
entanglement cost.

\paragraph{A single-copy lower bound on $E_{c,\mathcal K{\rm P}}^{\exp(\dagger)}$.---}
Let \(P\) denote the projector onto the support of a bipartite state \(\rho\).
Consider a Hermitian operator \(Q\) satisfying
\begin{equation}
  PQ=QP=0,
  \qquad
  \lVert Q\rVert_\infty\leq\frac{1}{2},
  \label{eq:H-constraints}
\end{equation}
and define
\begin{equation}
  \alpha(Q)
  :=
  \left\lVert(P+Q)^\Gamma\right\rVert_\infty .
  \label{eq:alpha-H}
\end{equation}
Here and below, \(\Gamma\) denotes partial transposition on Bob's subsystem.
Optimizing over all feasible \(Q\) gives
\begin{equation}
  \alpha_\star(\rho)
  :=
  \inf_Q \alpha(Q),
  \label{eq:alpha-star}
\end{equation}
where $Q$ is understood to obey the constraints in Eq.~\eqref{eq:H-constraints}. This optimization is a semidefinite program and can therefore be performed efficiently.

Our first result states that, whenever \(\alpha_\star(\rho)<1\),
\begin{equation}
  E_{c,\mathcal K{\rm P}}^{\exp(\dagger)}(\rho)
  \geq
  -\log_2\alpha_\star(\rho),
  \label{eq:sdp-cost-bound}
\end{equation}
for \(\mathcal K\in\{\operatorname{SEP},\operatorname{PPT}\}\). More
explicitly, consider any sequence of positive, trace-preserving
\(\mathcal K\)-preserving formation protocols
\begin{equation}
  \Omega_n
  :=
  \Lambda_n(\Phi_{M_n}).
\end{equation}
If their asymptotic rate lies strictly below
\(-\log_2\alpha_\star(\rho)\), then there exists \(\gamma>0\) such that
\begin{equation}
  1-
  T\!\left(
    \Omega_n,\rho^{\otimes n}
  \right)
  \leq
  2^{-\gamma n}
  \label{eq:formation-exponential-converse}
\end{equation}
for all sufficiently large \(n\). Thus, below the threshold, the formation
error approaches one exponentially fast. The proof, including the
semidefinite formulation of Eq.~\eqref{eq:alpha-star}, is given in the
Supplemental Material. Note that transforming a separable state into Bell pairs at any nonzero asymptotic rate necessarily yields an error that approaches one exponentially fast with the number of copies. Strong-converse rates therefore mark the regime beyond which the transformation becomes exponentially indistinguishable, in terms of its error scaling, from approximating Bell states using separable states alone.

\paragraph{Resolution of the Lami--Regula conjecture.---}
For \(\omega_3\), the support projector is \(P=P_3-\Phi_3\). Choosing
\(Q=-\Phi_3/2\) gives
\begin{equation}
  \left\lVert(P+Q)^\Gamma\right\rVert_\infty
  =
  \left\lVert P_3-\frac{1}{2}F_3\right\rVert_\infty
  =
  \frac{1}{2},
  \label{eq:omega-alpha}
\end{equation}
where \(F_3\) is the qutrit swap operator. Eq.~\eqref{eq:sdp-cost-bound}
therefore yields
\(E_{c,\mathcal K{\rm P}}^{\exp(\dagger)}(\omega_3)\geq1\).
Since one Bell pair prepares \(\omega_3\) exactly, this bound is tight.

For distillation, define $\sigma_3:=P_3/3$. The operator inequality
\begin{equation}
  \omega_3
  =
  \frac{P_3-\Phi_3}{2}
  \leq
  \frac{P_3}{2}
  =
  \frac{3}{2}\sigma_3
  \label{eq:omega-operator-bound}
\end{equation}
implies
\begin{equation}
  \Lambda_n(\omega_3^{\otimes n})
  \leq
  \left(\frac{3}{2}\right)^n
  \Lambda_n(\sigma_3^{\otimes n})
  \label{eq:omega-map-bound}
\end{equation}
for every positive map $\Lambda_n$. Since $\sigma_3$ is separable and
$\Lambda_n$ is $\mathcal{K}$-preserving,
$\Lambda_n(\sigma_3^{\otimes n})\in\mathcal{K}$, where
$\mathcal{K}\in\{\operatorname{SEP},\operatorname{PPT}\}$. Moreover, every
normalized $\tau\in\mathcal{K}$ satisfies
$\operatorname{Tr}(\Phi_M\tau)\leq 1/M$. Therefore,
\begin{align}
  f_n
  &:=
  \operatorname{Tr}\!\left[
    \Phi_{M_n}\Lambda_n(\omega_3^{\otimes n})
  \right]
  \nonumber\\
  &\leq
  \left(\frac{3}{2}\right)^n
  \operatorname{Tr}\!\left[
    \Phi_{M_n}\Lambda_n(\sigma_3^{\otimes n})
  \right]
  \nonumber\\
  &\leq
  \frac{(3/2)^n}{M_n}.
  \label{eq:omega-distillation-bound}
\end{align}
Since $\Phi_{M_n}$ is pure, the Fuchs--van de Graaf inequality \cite{Fuchs1999} gives
$1-\varepsilon_n^d\leq\sqrt{f_n}$. Hence, if
\begin{equation}
  \liminf_{n\to\infty}
  \frac{\log_2 M_n}{n}
  >
  \log_2\frac{3}{2},
\end{equation}
then $1-\varepsilon_n^d$ decays exponentially with $n$, and the distillation
error approaches one exponentially fast.
 This proves, together with Eq.~\eqref{eq:lami-regula-gap},
\begin{equation}
  E_{d,\mathcal K{\rm P}}^{\exp(\dagger)}(\omega_3)
  =
  \log_2\frac{3}{2}
  <1=
  E_{c,\mathcal K{\rm P}}^{\exp(\dagger)}(\omega_3),
  \label{eq:omega-exp-gap}
\end{equation}
resolving the conjecture of Ref.~\cite{Lami2023NoSecondLaw} and strengthening it
to the exponential strong-converse regime.

\paragraph{Completely PPT-preserving operations.---}
The support-based criterion above is useful for low-rank states, but it becomes
trivial for full-rank inputs: if $\rho$ has full support, then $P=\mathds{1}$,
the condition $PQ=0$ forces $Q=0$, and consequently
$\|(P+Q)^\Gamma\|_\infty=1$. To obtain nontrivial bounds beyond this setting,
we now restrict the allowed transformations to the more structured class of
completely PPT-preserving operations (C-PPT-P). This class strictly contains LOCC and
admits semidefinite characterizations that remain effective for arbitrary,
including full-rank, input states
\cite{AudenaertPlenioEisert2003,WangDuan2017}.

The distillation and formation problems are controlled by two different
semidefinite quantities. For distillation, we use the Rains set and the
max-Rains quantity \cite{berta2018amortization},
\begin{align}
  \mathcal R
  &:=
  \left\{
    X\geq0:\|X^\Gamma\|_1\leq1
  \right\},
  \nonumber\\
  R_{\max}(\rho)
  &:=
  \inf_{\tau\in\mathcal R}
  D_{\max}(\rho\Vert\tau),
  \label{eq:max-rains}
\end{align}
where
$D_{\max}(\rho\Vert\tau)
:=\inf\{\lambda:\rho\leq2^\lambda\tau\}$.
For formation, we introduce
\begin{multline}
  \mathrm{PPT}_2
  :=
  \bigl\{
    X\geq0:\exists\,Y\geq0,\;
    -Y\leq X^\Gamma\leq Y,\\[-1mm]
    \|Y^\Gamma\|_1\leq1
  \bigr\},
  \label{eq:ppt2-cone}
\end{multline}
together with
\begin{equation}
  E_{N,\operatorname{PPT}_2}^{1/2}(\rho)
  :=
  -\log_2
  \sup_{\sigma\in\mathrm{PPT}_2}
  F(\rho,\sigma).
  \label{eq:B2-definition}
\end{equation}
Both quantities are efficiently computable by semidefinite programming, and
$E_{N,\operatorname{PPT}_2}^{1/2}$ is exactly additive under tensor products
\cite{Wang2025Faithful}.

In the supplemental material, we show 
\begin{equation}
  E_{d,\mathrm{C\text{-}PPT\text{-}P}}^{\exp(\dagger)}(\rho)
  \leq R_{\max}(\rho),
  \;
  E_{c,\mathrm{C\text{-}PPT\text{-}P}}^{\exp(\dagger)}(\rho)
  \geq E_{N,\operatorname{PPT}_2}^{1/2}(\rho).
  \label{eq:cpptp-rate-bounds}
\end{equation}
Therefore, the single-copy condition
\begin{equation}
  R_{\max}(\rho)<E_{N,\operatorname{PPT}_2}^{1/2}(\rho)
  \label{eq:cpptp-criterion}
\end{equation}
certifies exponential strong-converse irreversibility. Unlike the
support-based criterion, Eq.~\eqref{eq:cpptp-criterion} can be nontrivial for
full-rank states.
\paragraph{Weighted complete-bipartite antisymmetric states.---}
We now construct a general analytically solvable family satisfying
Eq.~\eqref{eq:cpptp-criterion}. Let $L$ and $R$ be disjoint sets and let
$\Tilde{P}=(p_{ij})_{i\in L,j\in R}$ be a probability matrix. We define
\begin{equation}
  \beta_{\Tilde{P}}
  :=
  \sum_{i\in L,j\in R}
  p_{ij}\,
  |a_{ij}\rangle\!\langle a_{ij}|,
  \qquad
  |a_{ij}\rangle
  :=
  \frac{|ij\rangle-|ji\rangle}{\sqrt2}.
  \label{eq:weighted-antisymmetric-family}
\end{equation}
Writing $\|\Tilde{P}\|_1=\operatorname{Tr}\sqrt{\Tilde{P}^\dagger \Tilde{P}}$ for the trace norm of $P$, we find
\begin{equation}
  R_{\max}(\beta_{\Tilde{P}})
  =
  \log_2\!\left(1+\|\Tilde{P}\|_1\right),
  \qquad
  E_{N,\operatorname{PPT}_2}^{1/2}(\beta_{\Tilde{P}})=1.
  \label{eq:weighted-antisymmetric-values}
\end{equation}
The proof and the corresponding semidefinite-program witnesses are given in
the Supplemental Material.
Additionally, each component $|a_{ij}\rangle$ can be prepared exactly from one Bell pair,
and shared randomness produces the mixture $\beta_{\Tilde{P}}$. Consequently,
\begin{multline}
  E_{d,\mathrm{C\text{-}PPT\text{-}P}}^{\exp(\dagger)}(\beta_{\Tilde{P}})
  \leq
  \log_2\!\left(1+\|\Tilde{P}\|_1\right)\\[-1mm]
  <1=
  E_{c,\mathrm{C\text{-}PPT\text{-}P}}^{\exp(\dagger)}(\beta_{\Tilde{P}}),
  \qquad
  \|\Tilde{P}\|_1<1.
  \label{eq:weighted-antisymmetric-gap}
\end{multline}
Thus, the family exhibits exponential strong-converse irreversibility whenever
$\|\Tilde{P}\|_1<1$ and it contains the rank-two $3\otimes3$ construction of Wang and
Duan~\cite{WangDuan2017}.

\paragraph{Discussion.---}
Our results establish that mixed-state entanglement remains irreversible far beyond the conventional vanishing-error regime. Even under the maximal class of non-entangling quantum channels, permitting any fixed error strictly below one does not restore reversibility. More strongly, for a state $\omega_3$,whenever a protocol operates beyond the relevant distillation or formation threshold, its error converges exponentially fast to one with the number of copies. This resolves the conjecture of Ref.~\cite{Lami2023NoSecondLaw} and shows that entanglement irreversibility is neither a consequence of the restricted structure of LOCC nor an artifact of requiring asymptotically perfect transformations.

Our main technical contribution is an SDP-computable lower bound on the exponential strong-converse entanglement cost. Using this bound, we construct explicit examples for which any attempt to restore reversibility necessarily incurs an error that approaches one exponentially fast in the number of copies.

We further investigate the setting of completely PPT-preserving (C-PPT-P) operations. In this case, we derive efficiently computable single-copy criteria for exponential strong-converse irreversibility and identify an analytically tractable family of antisymmetric states that exhibits this phenomenon. Together, these results provide new insight into the fundamental limitations of entanglement manipulation beyond the vanishing-error regime.

\paragraph{Acknowledgments.---}
\begin{acknowledgments}
Some technical details of the proofs presented here were developed with the help of
ChatGPT (OpenAI, GPT-5.6 Sol; accessed July 2026). The authors thoroughly checked all output.
This work has been supported by the Federal Ministry
of Research, Technology and Space (BMFTR Projects QR.N, Grant No.
16KIS2202 and QSolid, Grant
No. 13N16163). We also acknowledge financial support by Deutsche Forschungsgemeinschaft (DFG, German Research Foundation) under Germany’s Excellence Strategy
– Cluster of Excellence Matter and Light for Quantum Computing (ML4Q) EXC 2004/1 – 390534769.
\end{acknowledgments}

\bibliography{strong_converse_irreversibility}

@article{Vidal2002Computable,
  author  = {Vidal, Guifr\'e and Werner, R. F.},
  title   = {Computable measure of entanglement},
  journal = {Physical Review A},
  volume  = {65},
  pages   = {032314},
  year    = {2002},
  doi     = {10.1103/PhysRevA.65.032314}
}

@article{Lami2023NoSecondLaw,
  author  = {Lami, Ludovico and Regula, Bartosz},
  title   = {No Second Law of Entanglement Manipulation after All},
  journal = {Nature Physics},
  volume  = {19},
  pages   = {184--189},
  year    = {2023},
  doi     = {10.1038/s41567-022-01873-9},
  eprint  = {2111.02438},
  archivePrefix = {arXiv}
}

@article{Rains2001SDP,
  author  = {Rains, Eric M.},
  title   = {A Semidefinite Program for Distillable Entanglement},
  journal = {IEEE Transactions on Information Theory},
  volume  = {47},
  number  = {7},
  pages   = {2921--2933},
  year    = {2001},
  doi     = {10.1109/18.959270}
}

@article{Takagi2022YieldCost,
  author  = {Takagi, Ryuji and Regula, Bartosz and Wilde, Mark M.},
  title   = {One-Shot Yield-Cost Relations in General Quantum Resource Theories},
  journal = {PRX Quantum},
  volume  = {3},
  pages   = {010348},
  year    = {2022},
  doi     = {10.1103/PRXQuantum.3.010348},
  eprint  = {2110.02212},
  archivePrefix = {arXiv}
}

@article{Wang2025Faithful,
  author  = {Wang, Xin and Jing, Mingrui and Zhu, Chengkai},
  title   = {Computable and Faithful Lower Bound on Entanglement Cost},
  journal = {Physical Review Letters},
  volume  = {134},
  pages   = {190202},
  year    = {2025},
  doi     = {10.1103/PhysRevLett.134.190202},
  eprint  = {2311.10649},
  archivePrefix = {arXiv}
}

@article{EinsteinPodolskyRosen1935,
  author  = {Einstein, Albert and Podolsky, Boris and Rosen, Nathan},
  title   = {Can Quantum-Mechanical Description of Physical Reality Be Considered Complete?},
  journal = {Physical Review},
  volume  = {47},
  number  = {10},
  pages   = {777--780},
  year    = {1935},
  doi     = {10.1103/PhysRev.47.777}
}

@article{Schrodinger1935,
  author  = {Schr{\"o}dinger, Erwin},
  title   = {Discussion of Probability Relations between Separated Systems},
  journal = {Mathematical Proceedings of the Cambridge Philosophical Society},
  volume  = {31},
  number  = {4},
  pages   = {555--563},
  year    = {1935},
  doi     = {10.1017/S0305004100013554}
}

@article{Bell1964,
  author  = {Bell, John S.},
  title   = {On the {E}instein {P}odolsky {R}osen Paradox},
  journal = {Physics Physique Fizika},
  volume  = {1},
  number  = {3},
  pages   = {195--200},
  year    = {1964},
  doi     = {10.1103/PhysicsPhysiqueFizika.1.195}
}

@article{ClauserHorneShimonyHolt1969,
  author  = {Clauser, John F. and Horne, Michael A. and Shimony, Abner and Holt, Richard A.},
  title   = {Proposed Experiment to Test Local Hidden-Variable Theories},
  journal = {Physical Review Letters},
  volume  = {23},
  number  = {15},
  pages   = {880--884},
  year    = {1969},
  doi     = {10.1103/PhysRevLett.23.880}
}

@article{HorodeckiReview2009,
  author  = {Horodecki, Ryszard and Horodecki, Pawe{\l} and Horodecki, Micha{\l} and Horodecki, Karol},
  title   = {Quantum Entanglement},
  journal = {Reviews of Modern Physics},
  volume  = {81},
  number  = {2},
  pages   = {865--942},
  year    = {2009},
  doi     = {10.1103/RevModPhys.81.865}
}

@article{BennettWiesner1992,
  author  = {Bennett, Charles H. and Wiesner, Stephen J.},
  title   = {Communication via One- and Two-Particle Operators on {E}instein--{P}odolsky--{R}osen States},
  journal = {Physical Review Letters},
  volume  = {69},
  number  = {20},
  pages   = {2881--2884},
  year    = {1992},
  doi     = {10.1103/PhysRevLett.69.2881}
}

@article{BennettTeleportation1993,
  author  = {Bennett, Charles H. and Brassard, Gilles and Cr{\'e}peau, Claude and Jozsa, Richard and Peres, Asher and Wootters, William K.},
  title   = {Teleporting an Unknown Quantum State via Dual Classical and {E}instein--{P}odolsky--{R}osen Channels},
  journal = {Physical Review Letters},
  volume  = {70},
  number  = {13},
  pages   = {1895--1899},
  year    = {1993},
  doi     = {10.1103/PhysRevLett.70.1895}
}

@article{Ekert1991,
  author  = {Ekert, Artur K.},
  title   = {Quantum Cryptography Based on {B}ell's Theorem},
  journal = {Physical Review Letters},
  volume  = {67},
  number  = {6},
  pages   = {661--663},
  year    = {1991},
  doi     = {10.1103/PhysRevLett.67.661}
}

@article{BriegelDuerCiracZoller1998,
  author  = {Briegel, Hans-J{\"u}rgen and D{\"u}r, Wolfgang and Cirac, J. Ignacio and Zoller, Peter},
  title   = {Quantum Repeaters: The Role of Imperfect Local Operations in Quantum Communication},
  journal = {Physical Review Letters},
  volume  = {81},
  number  = {26},
  pages   = {5932--5935},
  year    = {1998},
  doi     = {10.1103/PhysRevLett.81.5932}
}

@article{RaussendorfBriegel2001,
  author  = {Raussendorf, Robert and Briegel, Hans J.},
  title   = {A One-Way Quantum Computer},
  journal = {Physical Review Letters},
  volume  = {86},
  number  = {22},
  pages   = {5188--5191},
  year    = {2001},
  doi     = {10.1103/PhysRevLett.86.5188}
}

@article{BennettDiVincenzoSmolinWootters1996,
  author  = {Bennett, Charles H. and DiVincenzo, David P. and Smolin, John A. and Wootters, William K.},
  title   = {Mixed-State Entanglement and Quantum Error Correction},
  journal = {Physical Review A},
  volume  = {54},
  number  = {5},
  pages   = {3824--3851},
  year    = {1996},
  doi     = {10.1103/PhysRevA.54.3824}
}

@article{ChitambarEtAl2014,
  author  = {Chitambar, Eric and Leung, Debbie and Man{\v{c}}inska, Laura and Ozols, Maris and Winter, Andreas},
  title   = {Everything You Always Wanted to Know about {LOCC} (But Were Afraid to Ask)},
  journal = {Communications in Mathematical Physics},
  volume  = {328},
  number  = {1},
  pages   = {303--326},
  year    = {2014},
  doi     = {10.1007/s00220-014-1953-9}
}

@article{ChitambarGour2019,
  author  = {Chitambar, Eric and Gour, Gilad},
  title   = {Quantum Resource Theories},
  journal = {Reviews of Modern Physics},
  volume  = {91},
  number  = {2},
  pages   = {025001},
  year    = {2019},
  doi     = {10.1103/RevModPhys.91.025001}
}

@article{BrandaoDatta2011,
  author  = {Brand{\~a}o, Fernando G. S. L. and Datta, Nilanjana},
  title   = {One-Shot Rates for Entanglement Manipulation under Non-Entangling Maps},
  journal = {IEEE Transactions on Information Theory},
  volume  = {57},
  number  = {3},
  pages   = {1754--1760},
  year    = {2011},
  doi     = {10.1109/TIT.2011.2104531}
}

@article{BennettBernsteinPopescuSchumacher1996,
  author  = {Bennett, Charles H. and Bernstein, Herbert J. and Popescu, Sandu and Schumacher, Benjamin},
  title   = {Concentrating Partial Entanglement by Local Operations},
  journal = {Physical Review A},
  volume  = {53},
  number  = {4},
  pages   = {2046--2052},
  year    = {1996},
  doi     = {10.1103/PhysRevA.53.2046}
}

@article{HaydenHorodeckiTerhal2001,
  author  = {Hayden, Patrick M. and Horodecki, Micha{\l} and Terhal, Barbara M.},
  title   = {The Asymptotic Entanglement Cost of Preparing a Quantum State},
  journal = {Journal of Physics A: Mathematical and General},
  volume  = {34},
  number  = {35},
  pages   = {6891--6898},
  year    = {2001},
  doi     = {10.1088/0305-4470/34/35/314}
}

@article{Horodecki1997,
  author  = {Horodecki, Pawe{\l}},
  title   = {Separability Criterion and Inseparable Mixed States with Positive Partial Transposition},
  journal = {Physics Letters A},
  volume  = {232},
  number  = {5},
  pages   = {333--339},
  year    = {1997},
  doi     = {10.1016/S0375-9601(97)00416-7}
}

@article{HorodeckiBoundEntanglement1998,
  author  = {Horodecki, Micha{\l} and Horodecki, Pawe{\l} and Horodecki, Ryszard},
  title   = {Mixed-State Entanglement and Distillation: Is There a ``Bound'' Entanglement in Nature?},
  journal = {Physical Review Letters},
  volume  = {80},
  number  = {24},
  pages   = {5239--5242},
  year    = {1998},
  doi     = {10.1103/PhysRevLett.80.5239}
}

@article{BennettUPB1999,
  author  = {Bennett, Charles H. and DiVincenzo, David P. and Mor, Tal and Shor, Peter W. and Smolin, John A. and Terhal, Barbara M.},
  title   = {Unextendible Product Bases and Bound Entanglement},
  journal = {Physical Review Letters},
  volume  = {82},
  number  = {26},
  pages   = {5385--5388},
  year    = {1999},
  doi     = {10.1103/PhysRevLett.82.5385}
}

@article{VidalCirac2001,
  author  = {Vidal, Guifr{\'e} and Cirac, J. Ignacio},
  title   = {Irreversibility in Asymptotic Manipulations of Entanglement},
  journal = {Physical Review Letters},
  volume  = {86},
  number  = {25},
  pages   = {5803--5806},
  year    = {2001},
  doi     = {10.1103/PhysRevLett.86.5803}
}

@article{YangHorodeckiHorodeckiSynak2005,
  author  = {Yang, Dong and Horodecki, Micha{\l} and Horodecki, Ryszard and Synak-Radtke, Barbara},
  title   = {Irreversibility for All Bound Entangled States},
  journal = {Physical Review Letters},
  volume  = {95},
  number  = {19},
  pages   = {190501},
  year    = {2005},
  doi     = {10.1103/PhysRevLett.95.190501}
}

@article{LiebYngvason1999,
  author  = {Lieb, Elliott H. and Yngvason, Jakob},
  title   = {The Physics and Mathematics of the Second Law of Thermodynamics},
  journal = {Physics Reports},
  volume  = {310},
  number  = {1},
  pages   = {1--96},
  year    = {1999},
  doi     = {10.1016/S0370-1573(98)00082-9}
}

@article{BrandaoEtAlThermal2013,
  author  = {Brand{\~a}o, Fernando G. S. L. and Horodecki, Micha{\l} and Oppenheim, Jonathan and Renes, Joseph M. and Spekkens, Robert W.},
  title   = {Resource Theory of Quantum States Out of Thermal Equilibrium},
  journal = {Physical Review Letters},
  volume  = {111},
  number  = {25},
  pages   = {250404},
  year    = {2013},
  doi     = {10.1103/PhysRevLett.111.250404}
}

@article{BrandaoEtAlSecondLaws2015,
  author  = {Brand{\~a}o, Fernando G. S. L. and Horodecki, Micha{\l} and Ng, Nelly H. Y. and Oppenheim, Jonathan and Wehner, Stephanie},
  title   = {The Second Laws of Quantum Thermodynamics},
  journal = {Proceedings of the National Academy of Sciences},
  volume  = {112},
  number  = {11},
  pages   = {3275--3279},
  year    = {2015},
  doi     = {10.1073/pnas.1411728112}
}

@article{AudenaertPlenioEisert2003,
  author  = {Audenaert, Koenraad M. R. and Plenio, Martin B. and Eisert, Jens},
  title   = {Entanglement Cost under Positive-Partial-Transpose-Preserving Operations},
  journal = {Physical Review Letters},
  volume  = {90},
  number  = {2},
  pages   = {027901},
  year    = {2003},
  doi     = {10.1103/PhysRevLett.90.027901}
}

@article{WangDuan2017,
  author  = {Wang, Xin and Duan, Runyao},
  title   = {Irreversibility of Asymptotic Entanglement Manipulation under Quantum Operations Completely Preserving Positivity of Partial Transpose},
  journal = {Physical Review Letters},
  volume  = {119},
  number  = {18},
  pages   = {180506},
  year    = {2017},
  doi     = {10.1103/PhysRevLett.119.180506}
}

@article{lami2021framework,
  title={Framework for resource quantification in infinite-dimensional general probabilistic theories},
  author={Lami, Ludovico and Regula, Bartosz and Takagi, Ryuji and Ferrari, Giovanni},
  journal={Physical Review A},
  volume={103},
  number={3},
  pages={032424},
  year={2021},
  publisher={APS},
  doi={10.1103/PhysRevA.103.032424}
}

@article{vidal1999robustness,
  title={Robustness of entanglement},
  author={Vidal, Guifr{\'e} and Tarrach, Rolf},
  journal={Physical Review A},
  volume={59},
  number={1},
  pages={141},
  year={1999},
  publisher={APS},
  doi={10.1103/PhysRevA.59.141}
}

@book{nielsen2010quantum,
  title={Quantum computation and quantum information},
  author={Nielsen, Michael A and Chuang, Isaac L},
  year={2010},
  publisher={Cambridge University Press},
  doi={10.1017/CBO9780511976667}
}

@book{boucheron2013concentration,
  title={Concentration inequalities: A nonasymptotic theory of independence},
  author={Boucheron, S and Lugosi, G and Massart, P},
  year={2013},
  publisher={Oxford Academic},
  place={Oxford},
  doi={10.1093/acprof:oso/9780199535255.001.0001}
}

@article{berta2018amortization,
  title={Amortization does not enhance the max-{R}ains information of a quantum channel},
  author={Berta, Mario and Wilde, Mark M},
  journal={New Journal of Physics},
  volume={20},
  number={5},
  pages={053044},
  year={2018},
  publisher={IOP Publishing},
  doi={10.1088/1367-2630/aac153}
}

@article{wang2016improved,
  title={Improved semidefinite programming upper bound on distillable entanglement},
  author={Wang, Xin and Duan, Runyao},
  journal={Physical Review A},
  volume={94},
  number={5},
  pages={050301},
  year={2016},
  publisher={APS},
  doi={10.1103/PhysRevA.94.050301}
}

@article{Fuchs1999,
  author    = {Christopher A. Fuchs and Jeroen van de Graaf},
  title     = {Cryptographic Distinguishability Measures for Quantum-Mechanical States},
  journal   = {IEEE Transactions on Information Theory},
  volume    = {45},
  number    = {4},
  pages     = {1216--1227},
  year      = {1999},
  doi       = {10.1109/18.761271}
}

\clearpage
\appendix
\onecolumngrid

\section*{Supplemental Material: Very Strong Irreversibility of Quantum Entanglement}

\section{Definitions}
\label{sec:app1}

\paragraph{Free cones and operations.}
We follow the notation of Ref.~\cite{Lami2023NoSecondLaw}. For a bipartite
system $A:B$, let $\Sep(A:B)$ denote the cone generated by separable states,
and let
\begin{equation}
  \PPT(A:B)
  :=
  \left\{
    X\geq0:X^\Gamma\geq0
  \right\}
\end{equation}
denote the cone generated by PPT states, where $\Gamma$ is the partial
transpose on the second subsystem.

Throughout this section,
\begin{equation}
  \KK\in\{\Sep,\PPT\}.
\end{equation}
We write $\KKP$ for the class of positive trace-preserving maps that preserve
the cone $\KK$. Thus, a map
$\Lambda:A B\rightarrow A'B'$ belongs to $\KKP$ if
\begin{equation}
  \Lambda\!\left[\KK(A:B)\right]
  \subseteq
  \KK(A':B').
\end{equation}
For $\KK=\Sep$, these maps are usually called non-entangling operations.
For $\KK=\PPT$, they form the maximal class of PPT-state-preserving
operations.

We use the trace distance
\begin{equation}
  T(\rho,\sigma)
  :=
  \frac{1}{2}\lVert\rho-\sigma\rVert_1,
  \qquad
  \lVert X\rVert_1
  :=
  \operatorname{Tr}\sqrt{X^\dagger X}.
\end{equation}
The maximally entangled state of Schmidt rank $M$ is denoted by
\begin{equation}
  \Phi_M
  :=
  \lvert\Phi_M\rangle\!\langle\Phi_M\rvert,
  \qquad
  \lvert\Phi_M\rangle
  :=
  \frac{1}{\sqrt{M}}
  \sum_{i=0}^{M-1}\lvert ii\rangle.
\end{equation}

\paragraph{Fixed-error distillation and formation.}
Consider a sequence of distillation protocols
$\widetilde{\Lambda}_n\in\KKP$ with outputs of Schmidt rank $M_n$. Its rate
and error are
\begin{equation}
  r_n
  :=
  \frac{1}{n}\log_2 M_n,
  \qquad
  \varepsilon_n^{d}
  :=
  T\!\left(
    \widetilde{\Lambda}_n(\rho^{\otimes n}),
    \Phi_{M_n}
  \right).
\end{equation}
Similarly, for a sequence of formation protocols
$\Lambda_n\in\KKP$, define
\begin{equation}
  \varepsilon_n^{c}
  :=
  T\!\left(
    \Lambda_n(\Phi_{M_n}),
    \rho^{\otimes n}
  \right).
\end{equation}

For a fixed error tolerance $0\leq\varepsilon<1$, the distillable
entanglement and entanglement cost are
\cite{Vidal2002Computable,Lami2023NoSecondLaw}
\begin{align}
  E_{d,\KKP}^{\varepsilon}(\rho)
  &:=
  \sup
  \left\{
    R:
    \begin{array}{l}
      \exists\,(\widetilde{\Lambda}_n,M_n)_n
      \text{ such that}\\
      \displaystyle
      \liminf_{n\to\infty}r_n\geq R,\quad
      \limsup_{n\to\infty}\varepsilon_n^{d}\leq\varepsilon
    \end{array}
  \right\},
  \\
  E_{c,\KKP}^{\varepsilon}(\rho)
  &:=
  \inf
  \left\{
    R:
    \begin{array}{l}
      \exists\,(\Lambda_n,M_n)_n
      \text{ such that}\\
      \displaystyle
      \limsup_{n\to\infty}r_n\leq R,\quad
      \limsup_{n\to\infty}\varepsilon_n^{c}\leq\varepsilon
    \end{array}
  \right\}.
  \label{eq:fixed-error-formation}
\end{align}
The usual vanishing-error quantities are recovered at $\varepsilon=0$:
\begin{equation}
  E_{d,\KKP}(\rho)
  =
  E_{d,\KKP}^{0}(\rho),
  \qquad
  E_{c,\KKP}(\rho)
  =
  E_{c,\KKP}^{0}(\rho).
\end{equation}

\paragraph{Strong-converse thresholds.}
The strong-converse distillation and formation thresholds are
\begin{equation}
  E_{d,\KKP}^{\dagger}(\rho)
  :=
  \sup_{0\leq\varepsilon<1}
  E_{d,\KKP}^{\varepsilon}(\rho),
  \qquad
  E_{c,\KKP}^{\dagger}(\rho)
  :=
  \inf_{0\leq\varepsilon<1}
  E_{c,\KKP}^{\varepsilon}(\rho).
\end{equation}
Thus, above $E_{d,\KKP}^{\dagger}(\rho)$ no distillation protocol can keep
its error bounded away from one. Likewise, below
$E_{c,\KKP}^{\dagger}(\rho)$ no formation protocol can keep its error bounded
away from one.

For every $0\leq\varepsilon<1$, these quantities satisfy
\begin{align}
  E_{d,\KKP}(\rho)
  &\leq
  E_{d,\KKP}^{\varepsilon}(\rho)
  \leq
  E_{d,\KKP}^{\dagger}(\rho)
  \leq
  E_{c,\KKP}^{\dagger}(\rho)
  \leq
  E_{c,\KKP}^{\varepsilon}(\rho)
  \leq
  E_{c,\KKP}(\rho),
  \label{eq:yield-cost-chain}
\end{align}
where the inequality between the two strong-converse thresholds follows from
the general yield--cost relation of Ref.~\cite{Takagi2022YieldCost}.

\paragraph{Exponential strong-converse thresholds.}
The strong-converse thresholds do not specify how rapidly the error approaches
one. We therefore introduce exponential versions.

The exponential strong-converse distillation threshold is
\begin{align}
  E_{d,\KKP}^{\exp(\dagger)}(\rho)
  :=
  \inf\Big\{
    R:\,
    &\text{for every }(\widetilde{\Lambda}_n,M_n)_n
    \text{ with }
    \liminf_{n\to\infty}r_n>R,
    \nonumber\\[-1mm]
    &\exists\,\gamma>0,\ n_0\in\mathbb{N}
    \text{ such that }
    1-\varepsilon_n^{d}
    \leq2^{-\gamma n}
    \quad
    \forall n\geq n_0
  \Big\}.
  \label{eq:exp-sc-distillation-definition}
\end{align}
Hence, at every rate above
$E_{d,\KKP}^{\exp(\dagger)}(\rho)$, the distillation error approaches one
exponentially fast.

Similarly, the exponential strong-converse formation threshold is
\begin{align}
  E_{c,\KKP}^{\exp(\dagger)}(\rho)
  :=
  \sup\Big\{
    R:\,
    &\text{for every }(\Lambda_n,M_n)_n
    \text{ with }
    \limsup_{n\to\infty}r_n<R,
    \nonumber\\[-1mm]
    &\exists\,\gamma>0,\ n_0\in\mathbb{N}
    \text{ such that }
    1-\varepsilon_n^{c}
    \leq2^{-\gamma n}
    \quad
    \forall n\geq n_0
  \Big\}.
  \label{eq:exp-sc-formation-definition}
\end{align}
Thus, at every rate below
$E_{c,\KKP}^{\exp(\dagger)}(\rho)$, the formation error approaches one
exponentially fast.

By definition,
\begin{equation}
  E_{d,\KKP}^{\dagger}(\rho)
  \leq
  E_{d,\KKP}^{\exp(\dagger)}(\rho),
  \qquad
  E_{c,\KKP}^{\exp(\dagger)}(\rho)
  \leq
  E_{c,\KKP}^{\dagger}(\rho).
  \label{eq:exp-sc-ordering}
\end{equation}

In Sec.~\ref{sec:app2}, we derive efficiently computable lower bounds on
$E_{c,\KKP}^{\exp(\dagger)}$ and apply them to the Lami--Regula state
$\omega_3$. In Sec.~\ref{sec:weighted-antisymmetric}, we instead consider the more restricted class of completely PPT-preserving channels, namely CPTP maps $\Lambda$ for which
$\Gamma\circ\Lambda\circ\Gamma$ is completely positive, and construct a broad
family exhibiting exponential strong-converse irreversibility. This class
should not be confused with the maximal class of PPT-state-preserving maps
introduced above.

\section{Very strong irreversibility under PPT-preserving and non-entangling operations}\label{sec:app2}

We show a lower bound on $\Eexpscc$ in terms of the following quantity that we denote  kernel-robust negativity:

\begin{definition}[kernel-robust tempered negativity]
    Let $\rho$ be a rank deficient bipartite quantum state on systems $A$ and $B$ of arbitrary local dimension. Let $P$ denote the projector onto the support of $\rho$.
    Then its kernel-robust tempered negativity is defined as 
    \begin{align}\label{eq:krn}
        N^\perp(\rho) := \inf_{Q=Q^\dagger}\left\{\Vert (P + Q)^\Gamma \Vert_\infty\,:\,PQ=QP=0, \Vert Q\Vert_\infty \leq \frac12\right\},
    \end{align}
    where $X^\Gamma$ denotes partial transposition w.r.t.~the $B$ system and $\Vert \cdot \Vert_\infty$ is the operator norm, i.e., the largest singular value.
\end{definition}
Intuitively, this quantity describes the amount of tempered negativity of a state (defined in Ref.~\cite{Lami2023NoSecondLaw}) that cannot be removed by adding bounded operators to its kernel. Notably, it can be calculated efficiently by a semidefinite program.

\begin{theorem}
    Let $\rho$ be a rank deficient bipartite quantum state. Then
    \begin{align}
        \Eexpscc(\rho) \geq -\log_2 N^\perp(\rho),
    \end{align}
    where $N^\perp(\rho)$ denotes the kernel-robust tempered negativity defined in Eq.~\eqref{eq:krn}.
\end{theorem}
Before proving the theorem, note that while $N^\perp(\rho)$ gives the best lower bound, any feasible choice of $Q$ in Eq.~\eqref{eq:krn} yields a proper (yet possible suboptimal) bound.
\begin{proof}
For better readability, we outsource one of the major technical steps to Lemma~\ref{lem:Wn} below.

Let $\rho$ act on spaces $A$ and $B$ and let $P$ denote the projector onto the support of $\rho$. Then,
\begin{enumerate}[label=(\alph*)]
    \item Choose some $Q$  in the feasible set of Eq.~\eqref{eq:krn}. If $\Vert (P+Q)^\Gamma \Vert_\infty \geq 1$, the bound is trivial. Otherwise, we show in Lemma~\ref{lem:Wn} below that for any such choice of $Q$ in the feasible set of Eq.~\eqref{eq:krn} and every $\delta > 0$, one can construct a sequence of operators $W_n$ on $A^nB^n$ such that
    \begin{align} 
        \text{i.) } & \Vert W_n^\Gamma  \Vert_\infty \leq 1, \label{eq:Wn_ub}\\
        \text{ii.) } & W_n \geq L_n P^{\otimes n} - c_n L_n (\one - P^{\otimes n}), \label{eq:Wn_lb_1}
    \end{align}
    where $\limsup_{n\rightarrow \infty} \frac1n\log_2(c_n) \leq -\tilde{\delta}$ for some $\tilde{\delta} > 0$, and $\lim_{n\rightarrow \infty} \frac1n \log_2(L_n) \geq -\log_2(\Vert (P+Q)^\Gamma \Vert_\infty) -\delta $.
    
    \item Since $\Vert W_n^\Gamma \Vert_\infty \leq 1$, they can be used as feasible points in the dual formulation of the standard robustness $R_\KK^s$, given by \cite{lami2021framework}
    \begin{align}\label{eq:robustnessdual}
        1+2R_\KK^s(\sigma) = \sup_{W} \left\{ \Tr(W\sigma)\,:\,\vert \Tr(W\tau)\vert \leq \Tr(\tau) \forall \tau \in \KK\ \right\}.
    \end{align}
    \item Take any sequence of positive trace-preserving and $\KK$-preserving maps $(\Lambda_n)_{n\in \mathds{N}}$ such that $\Lambda_n$ maps $m_n$ Bell states to a state $\Omega_n$ on $A^nB^n$ (i.e., $\Lambda_n(\Phi_2^{\otimes m_n}) = \Omega_n$). 

    Now, for fixed $n$, use Eq.~\eqref{eq:robustnessdual} and monotonicity of the standard robustness (for both choices of $\KK$) to upper bound
    \begin{align}
        \Tr(W_n \Omega_n) \leq 1+2R_\KK^s(\Omega_n) \leq 1 + 2R_\KK^s(\Phi_2^{\otimes m_n}) = 2^{m_n+1}-1,
    \end{align}
    using the known value of robustness of $n$ copies of Bell states \cite{vidal1999robustness}.
    
    Conversely, use the bound in Eq.~\eqref{eq:Wn_lb_1} to lower bound
    \begin{align}
        \Tr(W_n \Omega_n) \geq L_n p_n - c_nL_n(1-p_n)
    \end{align}
    with $p_n = \Tr(P^{\otimes n}\Omega_n)$.
    The two bounds combined imply
    \begin{align} \label{eq:pn_bound}
        p_n &\leq \frac{2^{m_n + 1} - 1 + c_nL_n}{L_n(1+c_n)} \nonumber \\
         &\leq 2 \frac{2^{m_n}}{L_n}+c_n,
    \end{align}
    where we used $1+c_n \geq 1$.
    \item Use the variational form of the trace distance to bound \cite{nielsen2010quantum}
    \begin{align}
        \tracedist{\Omega_n}{\rho^{\otimes n}} &= \sup_{0\leq X \leq \one} \Tr[X(\rho^{\otimes n} - \Omega_n)] \\
        &\geq \Tr(P^{\otimes n}\rho^{\otimes n})) - \Tr(P^{\otimes n} \Omega_n) \\
        &= 1-p_n\\
        &\geq 1-\left(2 \frac{2^{m_n}}{L_n}+c_n\right). \label{eq:tracedist_lb}
    \end{align}
    In the second line, we have used the feasible point $X=P^{\otimes n}$, in the third line we used that $P$ projects onto the support of $\rho$ and in the last line we have used Eq.~\eqref{eq:pn_bound}.
    \item Finally, from the discussion above the definition of $\Eexpscc$, it is clear that in order to establish a lower bound on $\Eexpscc$, we have to show that if the rate $R[(\Lambda_n)_{n\in \mathds{N}}] := \limsup_{n\rightarrow \infty} m_n/n$ is strictly smaller than $-\log_2(\Vert (P+Q)^\Gamma \Vert_\infty)$, then $\tracedist{\Omega_n}{\rho^{\otimes n}} \geq 1 - 2^{-\gamma n}$ for large enough $n$ and some $\gamma > 0$.
    
    To see this, we make use of the elementary inequality $\log_2(a+b)\leq \log_2(2\max(a,b))=1+\log_2(\max(a,b))$.  Consider the bracket in Eq.~\eqref{eq:tracedist_lb} and set $A=\max(2\frac{2^{m_n}}{L_n}, c_n)$. Consequently,
    \begin{align} \label{eq:exp_upperbound}
        \limsup_{n\rightarrow \infty} \frac1n\log_2(2\frac{2^{m_n}}{L_n}+c_n) \leq \limsup_{n\rightarrow \infty} \frac1n[1+\log_2(A)]=\limsup_{n\rightarrow\infty}\frac1n \log_2(A).
    \end{align}
    Note that due to the construction of the $W_n$, we have that
    \begin{align}
        \limsup_{n\rightarrow\infty}\frac1n \log_2(c_n) &\leq -\tilde{\delta}, \\
        \limsup_{n\rightarrow\infty}\frac1n \log_2(2\frac{2^{m_n}}{L_n}) &\leq \limsup_{n\rightarrow\infty}\frac1n + R[(\Lambda_n)_{n\in \mathds{N}}] - \liminf_{n\rightarrow\infty}\frac1n\log_2(L_n)\\
        &\leq R[(\Lambda_n)_{n\in \mathds{N}}] 
        + \log_2(\Vert (P+Q)^\Gamma \Vert_\infty) + \delta. \label{eq:rate_ln_delta}
    \end{align}
    for some $\tilde{\delta} > 0$ and any choice of $\delta > 0$.
    
    Finally, assume that the rate is strictly smaller, i.e., $R[(\Lambda_n)_{n\in\mathds{N}}] = -\log_2(\Vert(P+Q)^\Gamma\Vert_\infty) - \kappa$ for some $\kappa > 0$. Then, Eq.~\eqref{eq:rate_ln_delta} yields
    \begin{align}
        \limsup_{n\rightarrow\infty}\frac1n\log_2(2\frac{2^{m_n}}{L_n}) \leq -\kappa + \delta.
    \end{align}
    Choosing, for instance, $\delta =\kappa/2$ then allows to bound Eq.~\eqref{eq:exp_upperbound} by
    \begin{align}
        \limsup_{n\rightarrow \infty} \frac1n\log_2(2\frac{2^{m_n}}{L_n}+c_n) \leq \max(-\tilde{\delta}, -\frac{\kappa}{2}) < 0,
    \end{align}
    showing exponential convergence of $\tracedist{\Omega_n}{\rho^{\otimes n}}$ to $1$ with $\gamma < \min(\tilde{\delta}, \frac{\kappa}{2})$ for large enough $n$. 
\end{enumerate}

\end{proof}

It remains to show the existence of the witness operators $W_n$.

\begin{lemma}\label{lem:Wn}
    Let $P$ be a projector onto some subspace of a bipartite quantum system $AB$, let $Q$ be an hermitian operator such that $QP = PQ = 0$, $\Vert Q \Vert_\infty \leq \frac12$ and $\Vert (P+Q)^\Gamma\Vert_\infty < 1$. Furthermore, let $\delta > 0$. Then there exists a sequence of operators $W_n$ acting on $A^nB^n$, such that 
    \begin{align}
        \Vert W_n^\Gamma\Vert_\infty &\leq 1, \label{eq:Wn_opnorm}\\
        W_n & \geq L_n P^{\otimes n} - c_n L_n (\one - P^{\otimes n}), \label{eq:Wn_lb}
    \end{align}
    where $\limsup_{n\rightarrow \infty} \frac1n\log_2(c_n) \leq -\tilde{\delta}$ for some $\tilde{\delta} > 0$, and $\lim_{n\rightarrow \infty} \frac1n \log_2(L_n) \geq -\log_2(\Vert (P+Q)^\Gamma \Vert_\infty) -\delta $.
\end{lemma}
\begin{proof}
    The idea of the witness construction is to allow to perturb $P^\Gamma$ by $Q^\Gamma$ as much as possible while keeping the weight of the operator within the subspace $P^{\otimes n}$ as much as possible.

    To that end, choose some $0 < \eta \leq 1$ (we will fix it at the end of the proof) and set
    \begin{align}\label{eq:Z0Z1}
        Z_0 = P+Q, \quad \quad Z_1 = P - \frac{1-\eta}{\eta}Q,
    \end{align}
    such that $(1-\eta)Z_0 + \eta Z_1 = P$. Then, for any choice of $n$, expand the $n$-fold tensor product as
    \begin{align}\label{eq:Potimesn_Z}
        P^{\otimes n} = [(1-\eta)Z_0 + \eta Z_1]^{\otimes n} = \sum_{T\subset{\{1,\ldots,n}\}} \eta^{\vert T \vert} (1-\eta)^{n-\vert T \vert} Z_T,
    \end{align}
    where the subset $T$ in the sum indicates the sites of the $n$ tensor copies, where $\eta Z_1$ is chosen in the expansion. Consequently, $Z_T := \bigotimes_{i=1}^n Z_{t_i}$ where $t_i = 1$ if $i\in T$ and $0$ otherwise.

    The idea is now to truncate this sum just after its mean to allow for the application of the Chernoff bound for binomial distributions later \cite{boucheron2013concentration}. To that end, we choose
    \begin{align}
        m := \lfloor c \eta n \rfloor
    \end{align}
    for some constant $c > 1$ (constrained to $c\eta < 1$) such that for large $n$, $m$ denotes a value just above the mean of the binomial distribution of $\eta n$. To make things more precise, we can fix $c=1.1$.
    
    We then cut and renormalize the sum in Eq.~\eqref{eq:Potimesn_Z} to just include subsets up to size $m$, yielding
    \begin{align}\label{eq:Wtilde}
        \tilde{W}_n:=\frac1{q_n(\eta)}\sum_{\substack{T\subset\{1,\ldots,n\}\\\vert T\vert \leq m}} \eta^{\vert T\vert} (1-\eta)^{n-\vert T \vert} Z_T.
    \end{align}
    with
    \begin{align}
        q_n(\eta) := \sum_{\substack{T\subset\{1,\ldots,n\}\\\vert T\vert \leq m}} \eta^{\vert T\vert} (1-\eta)^{n-\vert T \vert}.
    \end{align}
    Due to the normalization and the identity $Z_i P = P$ (due to $QP=0$), we have
    \begin{align} \label{eq:WntildeP}
        \tilde{W}_n P^{\otimes n} = P^{\otimes n}.
    \end{align}
    Finally, we need to control the operator norm of the partial transposition. To that end, note that convexity and multiplicativity of the operator norm w.r.t.~the tensor product yields
    \begin{align}\label{eq:opnormWgamma}
        \Vert \tilde{W}_n^\Gamma\Vert_\infty \leq \frac1{q_n}\sum_{\substack{T\subset\{1,\ldots,n\}\\\vert T\vert \leq m}} \eta^{\vert T\vert} (1-\eta)^{n-\vert T\vert} \Vert Z_0^\Gamma \Vert_\infty^{n-\vert T\vert} \Vert Z_1^\Gamma \Vert_\infty^{\vert T\vert}.
    \end{align}
    Let us set 
    \begin{align}
        \alpha &= \Vert Z_0^\Gamma \Vert_\infty = \Vert (P+Q)^\Gamma \Vert_\infty,\\
        \beta &= \max(\Vert Z_0^\Gamma \Vert_\infty, \Vert Z_1^\Gamma \Vert_\infty) \geq \alpha
    \end{align}
    and further upper bound Eq.~\eqref{eq:opnormWgamma} via
    \begin{align}
        \Vert \tilde{W}_n^\Gamma\Vert_\infty &\leq \frac1{q_n}\sum_{\substack{T\subset\{1,\ldots,n\}\\\vert T\vert \leq m}} \eta^{\vert T\vert} (1-\eta)^{n-\vert T\vert} \alpha^{n-\vert T\vert} \beta^{\vert T\vert} \\\
        &\leq \alpha^{n-m}\beta^m. \label{eq:Wtildegammabound}
    \end{align}
    Note that $\alpha <1$ by assumption. Here, we used that due to the normalization, the prefactors $\frac{\eta^{\vert T \vert} (1-\eta)^{n-\vert T\vert}}{q_n(\eta)}$ denote a probability distribution, and the largest of the coefficients $\alpha^{n-\vert T\vert} \beta^{\vert T\vert}$ is given whenever $\vert T\vert = m$.
    We finally rescale the witness to bound its partial transposition:
    \begin{align}\label{eq:Wn_def}
        W_n := \alpha^{m-n} \beta^{-m} \tilde{W}_n \equiv L_n \tilde{W}_n
    \end{align}
    with $L_n = \alpha^{m-n} \beta^{-m}$.
    Eq.~\eqref{eq:Wtildegammabound} then guarantees that $\Vert W_n^\Gamma \Vert_\infty \leq 1$, and Eq.~\eqref{eq:WntildeP} implies $W_n P^{\otimes n} = L_n P^{\otimes n}$.

    It remains to establish the claimed operator bound of $W_n$ together with the scaling. To that end, we rewrite the truncated sum in Eq.~\eqref{eq:Wtilde} in terms of tensor products of the operators $P$ and $Q$ by inserting Eq.~\eqref{eq:Z0Z1}. To do so, note that the expansion will contain any combination of $n$ tensor factors of $P$ and $Q$. We start with the original expression of $\tilde{W}_n$ and add the factor of $L_n$ at the end.
    
    Consider first the factor $P^{\otimes (n-s)} \otimes Q^{\otimes s}$ for some fixed $0\leq s \leq n$. Any term in the truncated sum labeled by some subset $T$ contributes to this factor exactly once. However, it will introduce an additional factor of $-(1-\eta)/\eta$ for each site $i\in T$ with $i > n-s$ where we look for factor $Q$ to be chosen. Of all subsets $T$ of size $\vert T\vert = j$ (with $0 \leq j \leq m$), there are exactly $ \binom{n-s}{j-r}\binom{s}{r}$ sets which add $r$ such factors. Thus, we obtain a prefactor (including the factors of $\eta$ and $1-\eta$ already present in Eq.~\eqref{eq:Wtilde}) of
    \begin{align} \label{eq:c_nms}
        c_{n,m,s}(\eta) &= \sum_{j=0}^m \sum_{r=0}^j \binom{n-s}{j-r}\binom{s}{r} \eta^j(1-\eta)^{n-j} \left[-\frac{1-\eta}{\eta}\right]^r\\
        &= \sum_{j=0}^m \sum_{r=0}^j\binom{n-s}{j-r}\binom{s}{r} (-1)^r \eta^{j-r}(1-\eta)^{n-(j-r)}.
    \end{align}
    This prefactor is the same for any permutation of $P^{\otimes (n-s)} \otimes Q^{\otimes s}$, and we obtain
    \begin{align}
        \tilde{W}_n = \frac1{q_n}\sum_{s=0}^n c_{n,m,s}(\eta) \mathcal{P}(P^{\otimes(n-s)} \otimes Q^{\otimes s}),
    \end{align}
    where $\mathcal{P}(\cdot)$ denotes the sum of all permutations of its argument, e.g., $\mathcal{P}(P^{\otimes 2} \otimes Q) = P\otimes P \otimes Q + P\otimes Q \otimes P + Q\otimes P \otimes P$.
    Note that all of the occurring tensor factors are mutually orthogonal since $PQ = QP = 0$. This implies that the operator norm of $\tilde{W}_n$ is bounded by
    \begin{align}
        \Vert \tilde{W}_n \Vert_\infty &= \max_{0\leq s \leq n} \frac{c_{n,m,s}(\eta)}{q_n} \Vert P^{\otimes (n-s)}\otimes Q^{\otimes s}\Vert_\infty \\
        &\leq \max_{0\leq s \leq n} \frac{c_{n,m,s}(\eta)}{q_n} \frac1{2^s},
    \end{align}
    where we have used $\Vert Q \Vert_\infty \leq 1/2$.
    Indeed, for $s=0$, we recover Eq.~\eqref{eq:WntildeP} and the corresponding eigenvalue of $1$ by noticing that
    \begin{align}
        c_{n,m,0}(\eta) &= \sum_{j=0}^m \sum_{r=0}^j\binom{n}{j-r}\binom{0}{r} (-1)^r \eta^{j-r}(1-\eta)^{n-(j-r)}. \\
        &= \sum_{j=0}^m \binom{n}{j} \eta^{j}(1-\eta)^{n-j} \\
        &= q_n.
    \end{align}
    For the lower bound we aim to prove in Eq.~\eqref{eq:Wn_lb}, however, we have to bound the remaining eigenvalues, for which we need to find 
    \begin{align}
        \Vert \tilde{W}_n -P^{\otimes n}\Vert_\infty &= \max_{1\leq s \leq n} \frac{c_{n,m,s}(\eta)}{q_n} \Vert P^{\otimes (n-s)}\otimes Q^{\otimes s}\Vert_\infty \\
        &\leq \max_{1\leq s \leq n} \frac{c_{n,m,s}(\eta)}{q_n} \frac1{2^s}
    \end{align}
    instead. 

    To determine an upper bound on $\frac{c_{n,m,s}(\eta)}{q_n} \frac1{2^s}$, note that, since $m =\lfloor c\eta n \rfloor > \eta n$ for large enough $n$, we can use the Chernoff bound to conclude \cite{boucheron2013concentration}
    \begin{align}\label{eq:q_n_bound}
        q_n(\eta) = \text{Pr}(X_n \leq m;\eta) = 1-\text{Pr}(X_n > m;\eta) \geq 1-2^{-nD(m/n \Vert \eta)},
    \end{align}
    with $\text{Pr}(X_n \leq m;\eta)$ denoting the probability of obtaining at most $m$ times outcome \textit{head} when a coin is flipped $n$ times and the probability of obtaining \textit{head} is given by $\eta$, and the Kullback-Leibler divergence $D(k \Vert \eta) := k \log_2(k/\eta) + (1-k)\log_2((1-k)/(1-\eta))$. Since $D(m/n \Vert \eta) \geq 0$ with equality if and only if $m/n = \eta$, $q_n \rightarrow 1$ for large $n$. Thus, we focus our concentration on bounding $2^{-s}c_{n,m,s}(\eta)$. We start by rewriting this as the coefficient of some polynomial.

    To that end, consider the polynomial $r(x) := (1-x)^s(1+\frac{\eta}{1-\eta}x)^{n-s}$ for fixed $s$. Expanding the brackets yields
    \begin{align}
        r(x) &= \sum_{r=0}^s \binom{s}{r}(-1)^r x^r \sum_{t=0}^{n-s}\binom{n-s}{t} \eta^t (1-\eta)^{-t} x^t  \\
        &= \sum_{j=0}^n \sum_{r=0}^s \binom{n-s}{j-r}\binom{s}{r}(-1)^r \eta^{j-r} (1-\eta)^{-(j-r)} x^j
    \end{align}
    with the substitution $j:=r+t$.
    For $s\geq 1$, we can expand
    \begin{align}
        \tilde{r}(x) &:= (1-x)^{s-1}(1+\frac{\eta}{1-\eta}x)^{n-s} \\
        &= \frac1{1-x} r(x)\\
        &= \sum_{k=0}^\infty x^k r(x) \\
        &= \sum_{k=0}^\infty \sum_{j=0}^n \sum_{r=0}^s \binom{n-s}{j-r}\binom{s}{r}(-1)^r \eta^{j-r} (1-\eta)^{-(j-r)} x^{j+k} \\
        &= \sum_{m=0}^\infty \sum_{j=0}^{\min(m,n)} \sum_{r=0}^s \binom{n-s}{j-r}\binom{s}{r}(-1)^r \eta^{j-r} (1-\eta)^{-(j-r)} x^{m} \\
        &= (1-\eta)^{-n}\sum_{m=0}^\infty \sum_{j=0}^{\min(m,n)} \sum_{r=0}^j \binom{n-s}{j-r}\binom{s}{r}(-1)^r \eta^{j-r} (1-\eta)^{n-(j-r)} x^{m}\\
        &=(1-\eta)^{-n}\sum_{m=0}^\infty c_{n,m,s}(\eta)x^{m},\\
    \end{align}
    where we used the substitution $m := j+k$ in the fourth line and the fact that we can equivalently sum $r$ from $0$ to $j$ due to the binomial coefficients in the sixth line. Note that for $m > n$, we implicitly extend here the definition of the $c_{n,m,s}(\eta)$ from Eq.~\eqref{eq:c_nms} to sum $j$ from $0$ to $\min(m,n)$. However, for our purposes, we are only interested in the case $m \leq n$.
    Thus, we can write 
    \begin{align}
        c_{n,m,s}(\eta) = [x^m](1-\eta)^n(1-x)^{s-1}(1+\frac{\eta}{1-\eta}x)^{n-s},
    \end{align}
    where $[x^m]t(x)$ denotes the coefficient in front of $x^m$ in the polynomial $t(x)$.
    
    Now use for any $x>0$, 
    \begin{align}\label{eq:cbound}
         \vert c_{n,m,s}(\eta)\vert &\leq [x^m](1-\eta)^n(1+x)^{s-1}(1+\frac{\eta}{1-\eta}x)^{n-s} \\
        &\leq  (1-\eta)^n(1+x)^{s-1}(1+\frac{\eta}{1-\eta}x)^{n-s} x^{-m}.
    \end{align}
    In the first line, we made use of the fact that exchanging the factor $1-x$ by $1+x$ effectively removes the factor of $(-1)^r$ in the expansion, turning (for $x>0$) an alternating sum into a sum of non-negative terms. In the second line, we made use of the fact that a polynomial $t(x)$ with positive coefficients fulfills $t(x) = \sum_{k=0}^\infty a_k x^k \geq a_m x^m$ for each $m$ and $x>0$. 

    Thus, each choice of $x>0$ yields a proper upper bound on $c_{n,m,s}(\eta)$. Set $\rho := s/n$, such that $\frac1n \leq \rho \leq 1$. We evaluate this bound for two specific values of $x$, one of which is good for small $\rho$, while the other one is good for large $\rho$. 

    \begin{description}
        \item[Case 1 ($x=1$)] Setting $x = 1$ in Eq.~\eqref{eq:cbound} yields
        \begin{align}
            2^{-s}\vert c_{n,m,s}(\eta)\vert &\leq 2^{-1} (1-\eta)^n (1+\frac{\eta}{1-\eta})^{n-s}\\
            &= \frac12(1-\eta)^s.
        \end{align}
        Define $L := -\log_2(1-\eta) > 0$, we can then write
        \begin{align} \label{eq:cbound1}
            2^{-s}\vert c_{n,m,s}(\eta) \vert \leq \frac12 2^{-n\rho L},
        \end{align}
        which is a strong bound for large $\rho \sim 1$, where it shows exponentially decaying behavior, but bad for small $\rho \sim 1/n$, where it becomes constant. We thus consider another choice of $x$.

        \item[Case 2 {($x = c\eta(1-\eta) / [\eta(1-c\eta)]$)}] For this choice of $x$ (whose positivity is ensured due to the constraint $c\eta < 1$), we obtain
        \begin{align}
            2^{-s}\vert c_{n,m,s}(\eta)\vert \leq \frac{(1-2c\eta + c)^{s-1}(1-c\eta)}{2^s c^m} \left(\frac{1-\eta}{1-c\eta}\right)^{n-m}.
        \end{align}
        Now $c \eta n - 1 \leq m = \lfloor c\eta n\rfloor \leq c \eta n$, therefore  $n(1-c\eta)\leq n-m \leq n(1-c\eta) + 1 $. Thus, since $(1-\eta)/(1-c\eta) > 1$ and $c > 1$, 
        \begin{align}
            2^{-s}\vert c_{n,m,s}(\eta)\vert &\leq \frac{(1-2c\eta + c)^{s-1}(1-c\eta)}{2^s c^{c\eta n-1}} \left(\frac{1-\eta}{1-c\eta}\right)^{n(1-c\eta)+1} \\
            &=  \frac{(1-2c\eta + c)^{s-1}(1-c\eta)(1-\eta)c}{2^s (1-c\eta)c^{c\eta n}} \left(\frac{1-\eta}{1-c\eta}\right)^{n(1-c\eta)}.
        \end{align}
        Finally, we use 
        \begin{align}
            \frac{1}{c^{c\eta n}}\left(\frac{1-\eta}{1-c\eta}\right)^{n(1-c\eta)} = 2^{-nD(c\eta \Vert \eta)}
        \end{align}
        and set 
        \begin{align}
            R_{c,\eta} &:= \frac{1-2c\eta + c}{2},\\
            C_{c,\eta} &:= \frac{(1-\eta)c}{1-2c\eta+c}.
        \end{align}
        Thus, 
        \begin{align}
            2^{-s}\vert c_{n,m,s}(\eta)\vert & \leq  C_{c,\eta} R_{c,\eta}^s 2^{-nD(c\eta \Vert \eta)}. 
        \end{align}
        This bound is good for small $s \sim \frac1n$. To see this, write $M := \log_2(R_{c,\eta})$. For small $\eta \leq \frac12 - \frac1{2c} = \frac1{22}$ (for the arbitrary choice of $c=1.1$), $R_{c,\eta} \geq 1$ and $M \geq 0$. Then,
        \begin{align} \label{eq:cbound2}
            2^{-s}\vert c_{n,m,s}(\eta)\vert \leq C_{c,\eta} 2^{-n(D(c\eta \Vert \eta) - \rho M)}
        \end{align}
        for $\rho = \frac sn$ and $\eta \leq \frac12 - \frac1{2c} = \frac1{22}$.
    \end{description}
    Let us combine both bounds from Eqs.~\eqref{eq:cbound1} and \eqref{eq:cbound2}. Under the assumption of $\eta \leq \frac12 - \frac1{2c}$, we have that $C_{c,\eta} \geq \frac12$. Thus,
    \begin{align} \label{eq:cbound3}
        2^{-s}\vert c_{n,m,s}(\eta)\vert &\leq \min(\frac12 2^{-n\rho L},C_{c,\eta} 2^{-n(D(c\eta \Vert \eta) - \rho M)}) \\
        &\leq C_{c,\eta} 2^{-n \max(\rho L, D(c\eta \Vert \eta) - \rho M)}.
    \end{align}
    To make the maximum of the two linear functions $\rho L$ and $D(c\eta \Vert \eta) - \rho M$ independent of the precise value of $\rho$, we can lower bound it at their intersection at $\rho = D(c\eta \Vert \eta) / (L+M)$, yielding finally
    \begin{align} \label{eq:cbound_combined}
        2^{-s}\vert c_{n,m,s}(\eta)\vert &\leq C_{c,\eta} 2^{-n \frac{D(c\eta \Vert \eta)L}{L+M}} \\
        &\leq 2^{-n\kappa(c,\eta)},
    \end{align}
    where $\kappa(c,\eta) := \frac{D(c\eta \Vert \eta)L}{L+M}$, $L = -\log_2(1-\eta) > 0$, $M = \log_2(R_{c,\eta}) > 0$, and we have used the fact that $C_{c\eta} \leq c/2 \leq 1$. Thus, the exponent is strictly negative and linear in $n$, showing exponential decay for any $1\leq s\leq n$ and $0 < \eta \leq \frac12 - \frac1{2c} = \frac1{22}$.

    Since Eq.~\eqref{eq:cbound_combined} upper bounds any eigenvalue of $\tilde{W}_n$ outside of the space $P^{\otimes n}$, we obtain (together with Eq.~\eqref{eq:WntildeP}) the operator inequality
    \begin{align}
        \tilde{W}_n \geq P^{\otimes n} - \frac{1}{q_n} 2^{-n\kappa(c,\eta)}(\one - P^{\otimes n}).
    \end{align}
    Multiplying by $L_n = \alpha^{m-n}\beta^{-m}$ from Eq.~\eqref{eq:Wn_def} yields
    \begin{align}
        W_n \geq L_nP^{\otimes n} - c_nL_n(\one - P^{\otimes n}).
    \end{align}
    with $c_n =2^{-n\kappa(c,\eta)}/q_n$.
    Now let us check the convergence. For $c_n$,
    \begin{align}
        \limsup_{n\rightarrow \infty} \frac1n \log_2(c_n) &= \limsup_{n\rightarrow \infty} \left[- \kappa(c,\eta) - \frac1n \log_2(q_n)\right] \\
        &= - \kappa(c,\eta) < 0.
    \end{align}
    Here, we have used that $q_n \rightarrow 1$ as shown in Eq.~\eqref{eq:q_n_bound}.
    Thus, setting $\tilde{\delta} = \kappa(c,\eta)$ yields the claimed behavior.

    Conversely, let $\delta>0$ be arbitrary. Then 
    \begin{align}
        \lim_{n\rightarrow \infty} \frac1n \log_2(L_n) &= \lim_{n\rightarrow \infty} \left[-(1-\frac mn)\log_2(\alpha) - \frac mn\log_2(\beta)\right] \\
        &= - \log_2(\alpha)-c\eta \log_2(\frac \beta \alpha),
    \end{align}
    since $m/n \rightarrow c\eta$ in the limit. Recall that $\alpha < 1$ and $\beta = \max(\alpha, \Vert (P-\frac{1-\eta}{\eta}Q)^\Gamma\Vert_\infty) \geq \alpha$.  
    For small enough $\eta$, $\Vert( P-\frac{1-\eta}{\eta}Q)^\Gamma\Vert_\infty$ is dominated by the diverging perturbation, growing like $1/\eta$ and leading to $\beta > \alpha$. Still, $\log_2(\beta)$ only grows logarithmically in $\eta$, such that $0<c\eta \log_2(\frac \beta \alpha) \rightarrow 0$ for $\eta \rightarrow 0$. Thus, for any choice of $\delta$, we can choose $\eta$ (possibly very small) such that $ - \log_2(\alpha)-c\eta \log_2(\frac \beta \alpha) \geq - \log_2(\alpha) - \delta$. This concludes the proof.
\end{proof}

\section{Very strong irreversibility under C-PPT preserving maps and family of antisymmetric states with analytical irreversiblity certificate}
\label{sec:weighted-antisymmetric}

\subsection{Upper and lower bounds on exponential strong-converse quantities for C-PPT preserving maps}

The Rains set on a bipartite system $A:B$ is \cite{Rains2001SDP}
\begin{equation}
  \mathcal{R}(A:B)
  := \left\{X\geq 0:
  \left\lVert X^{\Gamma_B}\right\rVert_1\leq 1\right\},
\end{equation}
and the max--Rains quantity is \cite{berta2018amortization}
\begin{align}
  R_{\max}(\rho)
  &:= \inf_{\tau\in\mathcal{R}(A:B)}
      D_{\max}(\rho\Vert\tau), \label{eq:Rmax} \\
  D_{\max}(\rho\Vert\tau)
  &:= \inf\left\{\lambda\in\mathbb{R}:
      \rho\leq 2^\lambda\tau\right\}.
\end{align}
Furthermore define \cite{Wang2025Faithful}
\begin{equation}\label{eq:PPT2}
  \operatorname{PPT}_2(A:B)
  := \left\{X\geq 0:
  \begin{array}{l}
    \exists \; Y\geq 0\text{ such that }
    -Y\leq X^{\Gamma_B}\leq Y,
    \left\lVert Y^{\Gamma_B}\right\rVert_1\leq 1
  \end{array}
  \right\},
\end{equation}
and
\begin{equation}\label{eq:EN12}
  E_{N,\operatorname{PPT}_2}^{1/2}(\rho)
  := -\log\sup_{\sigma\in\operatorname{PPT}_2(A:B)}
  F(\rho,\sigma).
\end{equation}
We use the exact additivity property
\begin{equation}
 E_{N,\operatorname{PPT}_2}^{1/2}\!\left(\rho^{\otimes n}\right)=nE_{N,\operatorname{PPT}_2}^{1/2}(\rho),
  \qquad n\in\mathbb{N},
  \label{eq:B2-additivity}
\end{equation}
proved in~\cite{Wang2025Faithful}.

A channel
$\Lambda:A B\to A'B'$ is called completely PPT-preserving (C-PPT preserving) iff
\begin{equation}
  \widetilde{\Lambda}
  := \Gamma_{B'}\circ\Lambda\circ\Gamma_B
\end{equation}
is completely positive.  Since all three maps in the composition are
trace preserving, $\widetilde{\Lambda}$ is then a channel. Every LOCC channel is completely PPT-preserving \cite{wang2016improved}.

\begin{lemma}[Preservation of the comparison sets]
\label{lem:preservation}
Every completely PPT-preserving channel maps
$\mathcal{R}(A:B)$ into $\mathcal{R}(A':B')$ and
$\operatorname{PPT}_2(A:B)$ into
$\operatorname{PPT}_2(A':B')$.
\end{lemma}

\begin{proof}
Let $\Lambda$ be completely PPT-preserving and let
$\widetilde{\Lambda}=\Gamma_{B'}\circ\Lambda\circ\Gamma_B$.
Because a channel contracts the trace norm on Hermitian operators, for
$X\in\mathcal{R}(A:B)$ we have
\begin{equation}
  \left\lVert \Lambda(X)^{\Gamma_{B'}}\right\rVert_1
  = \left\lVert \widetilde{\Lambda}
      \bigl(X^{\Gamma_B}\bigr)\right\rVert_1
  \leq \left\lVert X^{\Gamma_B}\right\rVert_1
  \leq 1.
\end{equation}
Hence $\Lambda(X)\in\mathcal{R}(A':B')$.

Now let $X\in\operatorname{PPT}_2(A:B)$, and choose a witness $Y\geq0$
such that
\begin{equation}
  -Y\leq X^{\Gamma_B}\leq Y,
  \qquad
  \left\lVert Y^{\Gamma_B}\right\rVert_1\leq1.
\end{equation}
Complete positivity of $\widetilde{\Lambda}$ gives
\begin{equation}
  -\widetilde{\Lambda}(Y)
  \leq \Lambda(X)^{\Gamma_{B'}}
  \leq \widetilde{\Lambda}(Y),
  \qquad
  \widetilde{\Lambda}(Y)\geq0.
\end{equation}
Moreover,
\begin{equation}
  \left\lVert
    \widetilde{\Lambda}(Y)^{\Gamma_{B'}}
  \right\rVert_1
  = \left\lVert \Lambda\bigl(Y^{\Gamma_B}\bigr)\right\rVert_1
  \leq \left\lVert Y^{\Gamma_B}\right\rVert_1
  \leq1.
\end{equation}
Thus $\widetilde{\Lambda}(Y)$ is a feasible witness for
$\Lambda(X)$, proving the second claim.
\end{proof}

Since the Rains set is tensor stable, subadditivity of the max-Rains quantity follows
\begin{equation}
  R_{\max}\!\left(\rho^{\otimes n}\right)
  \leq nR_{\max}(\rho).
  \label{eq:Rmax-subadditivity}
\end{equation}

Let $\mathcal{C}$ be a set of positive semidefinite operators, with the
underlying input and output spaces understood from context.  Define
\begin{align}
  R_{\max}^{\mathcal{C}}(\omega)
  &:= \inf_{\tau\in\mathcal{C}}
      D_{\max}(\omega\Vert\tau),\\
  E_{N,\mathcal{C}}^{1/2}(\sigma)
  &:= -\log\sup_{\eta\in\mathcal{C}}F(\sigma,\eta).
\end{align}

\begin{lemma}[One-shot conversion bound]
\label{lem:one-shot-conversion}
Let $\Lambda$ be a channel that preserves $\mathcal{C}$.  Then, for all
positive semidefinite $\omega$ and $\sigma$,
\begin{equation}
  F\bigl(\Lambda(\omega),\sigma\bigr)
  \leq
  2^{R_{\max}^{\mathcal{C}}(\omega)
     -E_{N,\mathcal{C}}^{1/2}(\sigma)}.
  \label{eq:one-shot-conversion}
\end{equation}
\end{lemma}

\begin{proof}
Fix $\varepsilon>0$.  By the definition of
$R_{\max}^{\mathcal{C}}(\omega)$, there exists
$\tau\in\mathcal{C}$ such that
\begin{equation}
  \omega
  \leq
  2^{R_{\max}^{\mathcal{C}}(\omega)+\varepsilon}\tau.
\end{equation}
Applying the completely positive map $\Lambda$ preserves this operator
inequality, and $\Lambda(\tau)\in\mathcal{C}$.  The squared fidelity is
monotone under the operator order in either positive argument and is
homogeneous of degree one in each argument.  Therefore,
\begin{align}
  F\bigl(\Lambda(\omega),\sigma\bigr)
  &\leq
  2^{R_{\max}^{\mathcal{C}}(\omega)+\varepsilon}
  F\bigl(\Lambda(\tau),\sigma\bigr)\\
  &\leq
  2^{R_{\max}^{\mathcal{C}}(\omega)+\varepsilon
     -E_{N,\mathcal{C}}^{1/2}(\sigma)}.
\end{align}
Letting $\varepsilon\downarrow0$ proves
\eqref{eq:one-shot-conversion}.
\end{proof}

Let
\begin{equation}
  \Phi_M
  := \lvert\Phi_M\rangle\!\langle\Phi_M\rvert,
  \qquad
  \lvert\Phi_M\rangle
  := \frac{1}{\sqrt{M}}\sum_{i=1}^{M}\lvert i\rangle\lvert i\rangle,
\end{equation}
be a maximally entangled state of Schmidt rank $M$.

\begin{theorem}[Distillation and formation]
\label{thm:strong-converse-bounds}
For every $n\in\mathbb{N}$, every positive integer $M$, and every
$\Lambda_n\in \mathrm{C\text{-}PPT\text{-}P} $,
\begin{align}
  F\bigl(\Lambda_n(\rho^{\otimes n}),\Phi_M\bigr)
  &\leq 2^{nR_{\max}(\rho)-\log M},
  \label{eq:distillation-bound}\\
  F\bigl(\Lambda_n(\Phi_M),\rho^{\otimes n}\bigr)
  &\leq 2^{\log M- n E_{N,\operatorname{PPT}_2}^{1/2}(\rho)}.
  \label{eq:formation-bound}
\end{align}
Consequently,
\begin{equation}
  E_{d,\mathrm{C\text{-}PPT\text{-}P}}^{\exp(\dagger)}
  \leq R_{\max}(\rho),
  \qquad
 E_{c,\mathrm{C\text{-}PPT\text{-}P}}^{\exp(\dagger)}
  \geq E_{N,\operatorname{PPT}_2}^{1/2}(\rho).
  \label{eq:threshold-bounds}
\end{equation}
\end{theorem}

\begin{proof}
We first prove the distillation bound.  For every
$\eta\in\mathcal{R}$,
\begin{align}
  F(\Phi_M,\eta)
  &= \operatorname{Tr}(\Phi_M\eta)
  = \operatorname{Tr}
     \bigl(\Phi_M^{\Gamma}\eta^{\Gamma}\bigr)\leq
  \left\lVert\Phi_M^{\Gamma}\right\rVert_\infty
  \left\lVert\eta^{\Gamma}\right\rVert_1 \leq \frac{1}{M},
\end{align}
where we used
$\Phi_M^{\Gamma}=V_M/M$ and hence
$\lVert\Phi_M^{\Gamma}\rVert_\infty=1/M$; here
$V_M$ is the swap operator.  Equality is attained by
$\eta=\Phi_M/M$, because
\begin{equation}
  \left\lVert(\Phi_M/M)^{\Gamma}\right\rVert_1
  = \left\lVert V_M/M^2\right\rVert_1
  =1.
\end{equation}
Therefore,
\begin{equation}
  E_{N,\mathcal{R}}^{1/2}(\Phi_M)=\log M.
  \label{eq:Rains-fidelity-Phi}
\end{equation}
Applying Lemma~\ref{lem:one-shot-conversion} with
$\mathcal{C}=\mathcal{R}$, together with
\eqref{eq:Rmax-subadditivity}, gives
\begin{align}
  F\bigl(\Lambda_n(\rho^{\otimes n}),\Phi_M\bigr)
  &\leq
  2^{R_{\max}(\rho^{\otimes n})-
     E_{N,\mathcal{R}}^{1/2}(\Phi_M)}\\
  &\leq 2^{nR_{\max}(\rho)-\log M},
\end{align}
which is \eqref{eq:distillation-bound}.

For formation, observe that $\Phi_M/M\in\operatorname{PPT}_2$.  Indeed,
\begin{equation}
  (\Phi_M/M)^{\Gamma}=V_M/M^2,
\end{equation}
and $Y=I_{M^2}/M^2$ is a feasible witness because
\begin{equation}
  -Y\leq V_M/M^2\leq Y,
  \qquad
  \left\lVert Y^{\Gamma}\right\rVert_1=1.
\end{equation}
It follows that
\begin{equation}
  R^{\text{PPT}_2}_{\max}(\Phi_M)
  \leq D_{\max}(\Phi_M\Vert\Phi_M/M)
  =\log M.
  \label{eq:PPT2-max-Phi}
\end{equation}
By the exact additivity in \eqref{eq:B2-additivity}, Lemma~\ref{lem:one-shot-conversion}, now with
$\mathcal{C}=\operatorname{PPT}_2$, yields
\begin{align}
  F\bigl(\Lambda_n(\Phi_M),\rho^{\otimes n}\bigr)
  &\leq
  2^{R^{\text{PPT}_2}_{\max}(\Phi_M)
     -E_{N,\operatorname{PPT}_2}^{1/2}(\rho^{\otimes n})}\\
  &\leq 2^{\log M-nE_{N,\operatorname{PPT}_2}^{1/2}(\rho)},
\end{align}
which is \eqref{eq:formation-bound}.

Now consider $M = M_n$ and the asymptotic rate $R_n = \frac{\log M_n}{n}$. For distillation the bound becomes
\begin{align}
    F\bigl(\Lambda_n(\rho^{\otimes n}),\Phi_{M_n}\bigr) \leq 2^{-n(R_n - R_{\max}(\rho))}.
\end{align}
Thus, whenever the distillation rate is strictly larger than $R_{\max}(\rho)$, the fidelity converges to zero exponentially. Therefore the exponential strong-converse distillable-entanglement threshold cannot exceed $R_{\max}(\rho)$ and we get
\begin{align}
    E_{d,\mathrm{C\text{-}PPT\text{-}P}}^{\exp(\dagger)} \leq R_{\max}(\rho).
\end{align}
Analogously the bound on the exponential strong-converse entanglement cost becomes
\begin{align}
    E_{c,\mathrm{C\text{-}PPT\text{-}P}}^{\exp(\dagger)}
  \geq E_{N,\operatorname{PPT}_2}^{1/2}(\rho).
\end{align}
\end{proof}

Therefore, if
\begin{equation}
  R_{\max}(\rho)<E_{N,\operatorname{PPT}_2}^{1/2}(\rho),
\end{equation}
then
\begin{equation}
  E_{d,\mathrm{C\text{-}PPT\text{-}P}}^{\exp(\dagger)}(\rho)
  \leq R_{\max}(\rho)
  < E_{N,\operatorname{PPT}_2}^{1/2}(\rho)
  \leq E_{c,\mathrm{C\text{-}PPT\text{-}P}}^{\exp(\dagger)}(\rho),
\end{equation}
so there is an explicit gap between the exponential strong-converse cost and formation under C-PPT preserving (or LOCC) maps.

\subsection{Family of antisymmetric states with analytically computable gap}

We can now define the family of antisymmetric states, for which we calculate both bounds explicitly and certify irreversibility.  Let $L,R\subseteq\{0,\ldots,d-1\}$ be disjoint sets with
$|L|=a$ and $|R|=b$. Let
$\Tilde{P}=(p_{ij})_{i\in L,j\in R}$ be a probability matrix, i.e.,
\begin{equation}
  p_{ij}\geq 0,
  \qquad
  \sum_{i\in L,j\in R}p_{ij}=1.
\end{equation}
For every $i\in L$ and $j\in R$, define
\begin{equation}
  \lvert a_{ij}\rangle
  :=
  \frac{\lvert ij\rangle-\lvert ji\rangle}{\sqrt{2}},
  \qquad
  A_{ij}
  :=
  \lvert a_{ij}\rangle\!\langle a_{ij}\rvert,
\end{equation}
and consider the state
\begin{equation}
  \beta_{\Tilde{P}}
  :=
  \sum_{i\in L,j\in R}p_{ij}A_{ij}.
  \label{eq:supp-weighted-beta}
\end{equation}
We denote the trace-norm of $\Tilde{P}$ by
\begin{equation}
  \lVert \Tilde{P}\rVert_{1}
  :=
  \operatorname{Tr}\sqrt{\Tilde{P}^\dagger \Tilde{P}}
  =
  \sum_k s_k(\Tilde{P}),
  \label{eq:supp-schatten-one}
\end{equation}
where $s_k(\Tilde{P})$ are the singular values of $\Tilde{P}$.

For every probability matrix $\Tilde{P}$, recall from Eqs.~\eqref{eq:Rmax} and \eqref{eq:EN12} that 
\begin{equation}
  R_{\max}(\beta_{\Tilde{P}})
  =
  \log_2\!\left(1+\lVert \Tilde{P}\rVert_{1}\right),
  \qquad
  E_{N,\operatorname{PPT}_2}^{1/2}(\beta_{\Tilde{P}})=1.
  \label{eq:supp-weighted-values}
\end{equation}

We prove the two equalities separately.

\paragraph{Evaluation of the max-Rains quantity.}
Taking the partial transpose on the second subsystem gives
\begin{align}
  \beta_{\Tilde{P}}^\Gamma
  =
  \frac{1}{2}
  \sum_{i\in L,j\in R}p_{ij}
  \Big(
    &\lvert ij\rangle\!\langle ij\rvert
    +
    \lvert ji\rangle\!\langle ji\rvert
    \nonumber\\
    &-
    \lvert ii\rangle\!\langle jj\rvert
    -
    \lvert jj\rangle\!\langle ii\rvert
  \Big).
  \label{eq:supp-beta-partial-transpose}
\end{align}
The first two terms form a positive diagonal operator of trace one on the
subspace spanned by
$\{\lvert ij\rangle,\lvert ji\rangle:i\in L,j\in R\}$.

The remaining terms act on the diagonal subspace
\begin{equation}
  \operatorname{span}\{\lvert ii\rangle:i\in L\}
  \oplus
  \operatorname{span}\{\lvert jj\rangle:j\in R\}.
\end{equation}
In the corresponding ordered basis, this block is
\begin{equation}
  -\frac{1}{2}
  \begin{pmatrix}
    0 & \Tilde{P}\\
    \Tilde{P}^{\mathsf T} & 0
  \end{pmatrix}.
  \label{eq:supp-diagonal-block}
\end{equation}
Its nonzero eigenvalues are $\pm s_k(\Tilde{P})/2$. Therefore,
\begin{equation}
  \lVert\beta_{\Tilde{P}}^\Gamma\rVert_1
  =
  1+\lVert \Tilde{P}\rVert_{1}.
  \label{eq:supp-beta-trace-norm}
\end{equation}

Using the SDP representation
\begin{equation}
  2^{R_{\max}(\rho)}
  =
  \min_{X\geq\rho}\lVert X^\Gamma\rVert_1,
  \label{eq:supp-max-rains-sdp}
\end{equation}
the feasible choice $X=\beta_{\Tilde{P}}$ gives
\begin{equation}
  2^{R_{\max}(\beta_{\Tilde{P}})}
  \leq
  1+\lVert \Tilde{P}\rVert_{1}.
  \label{eq:supp-rains-upper}
\end{equation}

To prove the opposite inequality, let
\begin{equation}
  \Tilde{P}=U\Sigma V^{\mathsf T}
\end{equation}
be a singular-value decomposition of $\Tilde{P}$, and define
\begin{equation}
  Z:=UV^{\mathsf T}.
\end{equation}
Then
\begin{equation}
  \lVert Z\rVert_\infty=1,
  \qquad
  \operatorname{Tr}(Z^{\mathsf T}\Tilde{P})
  =
  \operatorname{Tr}\Sigma
  =
  \lVert \Tilde{P}\rVert_{1}.
  \label{eq:supp-polar-properties}
\end{equation}
In particular, $|Z_{ij}|\leq1$ for every $i,j$.

Define the Hermitian operator
\begin{align}
  Q
  :=
  \sum_{i\in L,j\in R}
  \Big(
    &\lvert ij\rangle\!\langle ij\rvert
    +
    \lvert ji\rangle\!\langle ji\rvert
    \nonumber\\
    &-
    Z_{ij}\lvert ii\rangle\!\langle jj\rvert
    -
    Z_{ij}\lvert jj\rangle\!\langle ii\rvert
  \Big).
  \label{eq:supp-Q-witness}
\end{align}
On the off-diagonal subspace, $Q$ is the identity. On the diagonal subspace,
it is represented by
\begin{equation}
  -
  \begin{pmatrix}
    0&Z\\
    Z^{\mathsf T}&0
  \end{pmatrix}.
\end{equation}
It follows that
\begin{equation}
  \lVert Q\rVert_\infty\leq1.
  \label{eq:supp-Q-norm}
\end{equation}

The partial transpose of $Q$ is
\begin{align}
  W:=Q^\Gamma
  =
  \sum_{i\in L,j\in R}
  \Big(
    &\lvert ij\rangle\!\langle ij\rvert
    +
    \lvert ji\rangle\!\langle ji\rvert
    \nonumber\\
    &-
    Z_{ij}\lvert ij\rangle\!\langle ji\rvert
    -
    Z_{ij}\lvert ji\rangle\!\langle ij\rvert
  \Big).
  \label{eq:supp-W-witness}
\end{align}
For every pair $(i,j)$, its restriction to
$\operatorname{span}\{\lvert ij\rangle,\lvert ji\rangle\}$ is
\begin{equation}
  \begin{pmatrix}
    1&-Z_{ij}\\
    -Z_{ij}&1
  \end{pmatrix},
\end{equation}
which is positive semidefinite because $|Z_{ij}|\leq1$. Hence
\begin{equation}
  W=Q^\Gamma\geq0.
  \label{eq:supp-W-positive}
\end{equation}

For every $X\geq\beta_{\Tilde{P}}$, we obtain (using Hölder's inequality)
\begin{align}
  \lVert X^\Gamma\rVert_1
  \geq
  \operatorname{Tr}(QX^\Gamma)
  \nonumber=
  \operatorname{Tr}(WX)
  \nonumber\geq
  \operatorname{Tr}(W\beta_{\Tilde{P}})
  \nonumber=
  \operatorname{Tr}(Q\beta_{\Tilde{P}}^\Gamma)
  \nonumber=
  1+\operatorname{Tr}(Z^{\mathsf T}\Tilde{P})
  \nonumber=
  1+\lVert \Tilde{P}\rVert_{1}.
  \label{eq:supp-rains-lower}
\end{align}
This inequality establishes through Eq.~\eqref{eq:supp-max-rains-sdp} together with Eq.~\eqref{eq:supp-rains-upper} that
\begin{equation}
  R_{\max}(\beta_{\Tilde{P}})
  =
  \log_2\!\left(1+\lVert \Tilde{P}\rVert_{1}\right).
\end{equation}

\paragraph{Evaluation of the $ E_{N,\operatorname{PPT}_2}^{1/2}$.} Let
\begin{equation}
  \Pi
  :=
  \sum_{i\in L,j\in R}A_{ij}
  \label{eq:supp-support-projector}
\end{equation}
be the projector onto the part of the antisymmetric subspace where $\beta_{\Tilde{P}}$ is supported. The
state $\beta_{\Tilde{P}}$ is supported on this subspace.

Define
\begin{equation}
  \lvert u_L\rangle
  :=
  \sum_{i\in L}\lvert ii\rangle,
  \qquad
  \lvert u_R\rangle
  :=
  \sum_{j\in R}\lvert jj\rangle,
\end{equation}
and
\begin{align}
  H
  :=
  \frac{1}{2}
  \Bigg[
    &\lvert u_L\rangle\!\langle u_L\rvert
    +
    \lvert u_R\rangle\!\langle u_R\rvert
    \nonumber\\
    &+
    \sum_{i\in L,j\in R}
    \left(
      \lvert ij\rangle\!\langle ij\rvert
      +
      \lvert ji\rangle\!\langle ji\rvert
    \right)
  \Bigg].
  \label{eq:supp-H-witness}
\end{align}
A direct calculation gives
\begin{equation}
  H-\Pi^\Gamma
  =
  \frac{1}{2}
  \big(
    \lvert u_L\rangle+\lvert u_R\rangle
  \big)
  \big(
    \langle u_L\rvert+\langle u_R\rvert
  \big)
  \geq0
  \label{eq:supp-H-minus}
\end{equation}
and
\begin{align}
  H+\Pi^\Gamma
  =
  &\sum_{i\in L,j\in R}
  \left(
    \lvert ij\rangle\!\langle ij\rvert
    +
    \lvert ji\rangle\!\langle ji\rvert
  \right)
  \nonumber\\
  &+
  \frac{1}{2}
  \big(
    \lvert u_L\rangle-\lvert u_R\rangle
  \big)
  \big(
    \langle u_L\rvert-\langle u_R\rvert
  \big)
  \geq0.
  \label{eq:supp-H-plus}
\end{align}
Moreover,
\begin{equation}
  H^\Gamma
  =
  \frac{1}{2}
  \left(
    V_L+V_R+
    \sum_{i\in L,j\in R}
    \left[
      \lvert ij\rangle\!\langle ij\rvert
      +
      \lvert ji\rangle\!\langle ji\rvert
    \right]
  \right),
  \label{eq:supp-H-partial-transpose}
\end{equation}
where
\begin{equation}
  V_L
  :=
  \sum_{i,i'\in L}
  \lvert ii'\rangle\!\langle i'i\rvert,
  \qquad
  V_R
  :=
  \sum_{j,j'\in R}
  \lvert jj'\rangle\!\langle j'j\rvert.
\end{equation}
The three terms in Eq.~\eqref{eq:supp-H-partial-transpose} act on mutually
orthogonal subspaces and each has operator norm one. Consequently,
\begin{equation}
  \lVert H^\Gamma\rVert_\infty=\frac{1}{2}.
  \label{eq:supp-H-norm}
\end{equation}

Let $\sigma\in\mathrm{PPT}_2$. By definition in Eq.~\eqref{eq:PPT2}, there exists $Y\geq0$ such that
\begin{equation}
  -Y\leq\sigma^\Gamma\leq Y,
  \qquad
  \lVert Y^\Gamma\rVert_1\leq1.
  \label{eq:supp-ppt2-feasible}
\end{equation}
Using Eqs.~\eqref{eq:supp-H-minus} and \eqref{eq:supp-H-plus}, we find
\begin{align}\label{eq:supp-ppt2-overlap-bound}
  &\operatorname{Tr}(HY)
  -
  \operatorname{Tr}(\Pi^\Gamma\sigma^\Gamma)
  \nonumber\\
  &=
  \frac{1}{2}
  \operatorname{Tr}
  \left[
    (H-\Pi^\Gamma)(Y+\sigma^\Gamma)
  \right]
  \nonumber\\
  &\quad+
  \frac{1}{2}
  \operatorname{Tr}
  \left[
    (H+\Pi^\Gamma)(Y-\sigma^\Gamma)
  \right]
  \geq0.
\end{align}
Therefore,
\begin{align}
  \operatorname{Tr}(\Pi\sigma)=
  \operatorname{Tr}(\Pi^\Gamma\sigma^\Gamma)
  \nonumber\leq
  \operatorname{Tr}(HY)
  \nonumber=
  \operatorname{Tr}(H^\Gamma Y^\Gamma)
  \nonumber\leq
  \lVert H^\Gamma\rVert_\infty
  \lVert Y^\Gamma\rVert_1
  \nonumber \leq
  \frac{1}{2}.
  \label{eq:supp-support-overlap}
\end{align}

Since $\beta_{\Tilde{P}}$ is supported on $\Pi$, monotonicity of the fidelity under the
two-outcome measurement $\{\Pi,\mathds{1}-\Pi\}$ gives
\begin{equation}
  F(\beta_{\Tilde{P}},\sigma)
  \leq
  \operatorname{Tr}(\Pi\sigma)
  \leq
  \frac{1}{2}.
  \label{eq:supp-B2-upper-fidelity}
\end{equation}
It follows that
\begin{equation}
  E_{N,\operatorname{PPT}_2}^{1/2}(\beta_{\Tilde{P}})\geq1.
  \label{eq:supp-B2-lower}
\end{equation}

For the reverse inequality, fix one edge $(i,j)$ and define
\begin{equation}
  X_{ij}:=\frac{1}{2}A_{ij}
\end{equation}
and
\begin{equation}
  Y_{ij}
  :=
  \frac{1}{4}
  \left(
    \lvert ii\rangle\!\langle ii\rvert
    +
    \lvert ij\rangle\!\langle ij\rvert
    +
    \lvert ji\rangle\!\langle ji\rvert
    +
    \lvert jj\rangle\!\langle jj\rvert
  \right).
  \label{eq:supp-edge-Y}
\end{equation}
A direct calculation shows that
\begin{equation}
  -Y_{ij}
  \leq
  X_{ij}^\Gamma
  \leq
  Y_{ij},
  \qquad
  \lVert Y_{ij}^\Gamma\rVert_1=1.
\end{equation}
Hence $A_{ij}/2\in\mathrm{PPT}_2$. By convexity,
\begin{equation}
  \frac{1}{2}\beta_{\Tilde{P}}
  =
  \sum_{i\in L,j\in R}
  p_{ij}\frac{A_{ij}}{2}
  \in
  \mathrm{PPT}_2.
  \label{eq:supp-beta-half-ppt2}
\end{equation}
Using homogeneity of the squared fidelity,
\begin{equation}
  F\left(\beta_{\Tilde{P}},\frac{1}{2}\beta_{\Tilde{P}}\right)
  =
  \frac{1}{2}.
\end{equation}
Therefore,
\begin{equation}
  E_{N,\operatorname{PPT}_2}^{1/2}(\beta_{\Tilde{P}})\leq1.
\end{equation}
Together with Eq.~\eqref{eq:supp-B2-lower}, this proves
\begin{equation}
  E_{N,\operatorname{PPT}_2}^{1/2}(\beta_{\Tilde{P}})=1.
\end{equation}

The one-shot conversion bounds derived in Theorem~\ref{thm:strong-converse-bounds} yield
\begin{equation}
  E_{d,\mathrm{C\text{-}PPT\text{-}P}}^{\exp(\dagger)}(\rho)
  \leq
  R_{\max}(\rho),
  \qquad
  E_{c,\mathrm{C\text{-}PPT\text{-}P}}^{\exp(\dagger)}(\rho)
  \geq
  E_{N,\operatorname{PPT}_2}^{1/2}(\rho).
\end{equation}
The bound therefore gives
\begin{equation}
  E_{d,\mathrm{C\text{-}PPT\text{-}P}}^{\exp(\dagger)}(\beta_{\Tilde{P}})
  \leq
  \log_2\!\left(1+\lVert \Tilde{P}\rVert_{1}\right)
\end{equation}
and
\begin{equation}
  E_{c,\mathrm{C\text{-}PPT\text{-}P}}^{\exp(\dagger)}(\beta_{\Tilde{P}})\geq1.
\end{equation}

Conversely, one Bell pair can be transformed exactly into any
$\lvert a_{ij}\rangle$ by local isometries. The parties can use shared
randomness to choose $(i,j)$ with probability $p_{ij}$, thereby preparing
$\beta_{\Tilde{P}}$ exactly. Hence
\begin{equation}
  E_{c,\mathrm{C\text{-}PPT\text{-}P}}^{\exp(\dagger)}(\beta_{\Tilde{P}})\leq1,
\end{equation}
which proves equality.

Finally, the triangle inequality for the Schatten $1$-norm gives
\begin{align}
  \lVert \Tilde{P}\rVert_{1}
  &=
  \left\lVert
    \sum_{i,j}p_{ij}
    \lvert i\rangle\!\langle j\rvert
  \right\rVert_{1}
  \nonumber\\
  &\leq
  \sum_{i,j}p_{ij}
  \left\lVert
    \lvert i\rangle\!\langle j\rvert
  \right\rVert_{1}
  =
  1.
\end{align}
Whenever this inequality is strict, the exponential strong-converse distillation threshold is strictly smaller than the corresponding formation
threshold, proving irreversibility. As a concrete example, consider
\begin{equation}
  \beta_{\star}
  :=
  \frac{1}{2}
  \left(
    |a_{01}\rangle\!\langle a_{01}|
    +
    |a_{02}\rangle\!\langle a_{02}|
  \right)
\end{equation}
on $\mathbb{C}^3\otimes\mathbb{C}^3$. In this case,
$\Tilde{P}=(1/2,1/2)$ and hence $\|\Tilde{P}\|_{1}=1/\sqrt{2}$. Consequently,
\begin{equation}
  E_{d,\mathrm{C\text{-}PPT\text{-}P}}^{\exp(\dagger)}(\beta_{\star})
  \leq
  \log_2\!\left(1+\frac{1}{\sqrt{2}}\right)
  <1
  =
  E_{c,\mathrm{C\text{-}PPT\text{-}P}}^{\exp(\dagger)}(\beta_{\star}).
\end{equation}


\end{document}